\documentclass[11pt]{article}
\usepackage[nomath,nothms,nomisc,nolayout,nocolor,nohyperref,abbrs]{mymacros}

\usepackage{amsfonts}
\usepackage{amsthm}
\usepackage{amscd}
\usepackage{stmaryrd}
\usepackage{amssymb}
\usepackage{graphics}
\usepackage{amsmath}
\usepackage[english]{babel}
\usepackage{color}
\usepackage[usenames,dvipsnames]{xcolor}
\usepackage[colorlinks=true,linkcolor=NavyBlue,urlcolor=RoyalBlue,citecolor=PineGreen,linktocpage=true]{hyperref}
\usepackage{hyperref}
\usepackage{bbm}
\usepackage{yfonts}
\usepackage{mathrsfs}
\usepackage{upgreek}
\usepackage{epsfig}
\usepackage{graphicx}
\usepackage{enumerate}
\usepackage{url}
\usepackage[T1]{fontenc}
\usepackage[latin9]{inputenc}
 \usepackage{cleveref}
\usepackage{breqn}
\usepackage{geometry}
\DeclareMathAlphabet{\mathpzc}{OT1}{pzc}{m}{it}

\numberwithin{equation}{section}

\usepackage{titlesec}
\titleformat{\subsection}[runin]{\normalfont\bfseries}{\thesubsection.}{.5em}{}[.]\titlespacing{\subsection}{0pt}{2ex plus .1ex minus .2ex}{.8em}
\titleformat{\subsubsection}[runin]{\normalfont\itshape}{\thesubsubsection.}{.3em}{}[.]\titlespacing{\subsubsection}{0pt}{1ex plus .1ex minus .2ex}{.5em}
\titleformat{\paragraph}[runin]{\normalfont\itshape}{\theparagraph.}{.3em}{}[.]\titlespacing{\paragraph}{0pt}{1ex plus .1ex minus .2ex}{.5em}

\newtheorem{theorem}{Theorem}[section]
\newtheorem{lemma}[theorem]{Lemma}
\newtheorem{corollary}[theorem]{Corollary}

\newtheorem{claim}[theorem]{Claim}

\theoremstyle{definition}

\theoremstyle{remark}
\newtheorem{remark}[theorem]{Remark}

\newcommand{\veiii}[1]{{\left\vert\kern-0.25ex\left\vert\kern-0.25ex\left\vert #1 
    \right\vert\kern-0.25ex\right\vert\kern-0.25ex\right\vert}}

\newcommand{\BC}{{\mathbb{C}}}

\newcommand{\BH}{{\mathbb{H}}}

\newcommand{\BN}{{\mathbb{N}}}

\newcommand{\BR}{{\mathbb{R}}}

\newcommand{\BZ}{{\mathbb{Z}}}

\newcommand{\CA}{{\mathcal{A}}}

\newcommand{\CC}{{\mathcal{C}}}
\renewcommand{\CD}{{\mathcal{D}}}
\newcommand{\CE}{{\mathcal{E}}}
\newcommand{\CF}{{\mathcal{F}}}
\newcommand{\CG}{{\mathcal{G}}}
\newcommand{\CH}{{\mathcal{H}}}

\newcommand{\CI}{{\mathcal{I}}}

\newcommand{\CP}{{\mathcal{P}}}
\newcommand{\CQ}{{\mathcal{Q}}}
\newcommand{\CR}{{\mathcal{R}}}
\newcommand{\CS}{{\mathcal{S}}}

\newcommand{\SG}{{\mathscr{G}}}

\usepackage{tikz}

\usepackage{ifthen}

\usetikzlibrary{calc}
\usetikzlibrary{petri}
\usepackage{pgf}
\usepackage{pgffor}

\usetikzlibrary{arrows,shapes,positioning}
\usetikzlibrary{decorations.markings}

\makeatother

\tikzstyle arrowstyle=[scale=1]
\tikzstyle directed=[postaction={decorate,decoration={markings,
    mark=at position .65 with {\arrow[arrowstyle]{stealth}}}}]
\tikzstyle reverse directed=[postaction={decorate,decoration={markings,
    mark=at position .65 with {\arrowreversed[arrowstyle]{stealth};}}}]
    \tikzstyle{placec}=[circle, draw=black,fill=white, inner sep=2pt]
     \tikzstyle{placer}=[rectangle, draw=black,fill=white, inner sep=3pt]

\newcommand{\RN}[1]{
  \textup{\uppercase\expandafter{\romannumeral#1}}
}

\newcommand{\eps}{\varepsilon}

\renewcommand{\H}{\mathbb{H}}

\newcommand{\n}{\mathrm{n}}
\newcommand{\ut}{\underline{t}}

\newcommand{\beq}{\begin{equation}}
\newcommand{\eeq}{\end{equation}}

\newcommand{\ssp}{\mathfrak{sp}}

\title{Supersymmetric Hyperbolic $\sigma$-models and \\
Decay of Correlations in Two Dimensions}
\date{December 11, 2019}
\author{Nicholas Crawford \thanks{supported by Israel Science Foundation grant number 1692/17}\\ Department of Mathematics, The Technion}

\begin{document}

\maketitle

\begin{abstract}
In this paper we study a family of nonlinear $\sigma$-models in which the target space is the super manifold $\H^{2|2N}$.  These models generalize Zirnbauer's $\H^{2|2}$ nonlinear $\sigma$-model \cite{Z}.  The latter model has a number of special features which aid in its analysis: by supersymmetric localization, the partition function of the $\H^{2|2}$ model is one independent of the coupling constants (!). Our main technical observation is to generalize this fact to $\H^{2|2N}$ models as follows: the partition function is a multivariate polynomial of degree $n=N-1$, increasing in each variable.  As an application, these facts provide estimates on the Fourier and Laplace transforms of  the '$t$-field' when we specialize to $\Z^2$. From the bounds, we conclude the $t$-field exhibits polynomial decay of correlations and also has fluctuations which are at least those of a massless free field.
\end{abstract}
\section{Introduction}
We begin by introducing one of the main objects of study in this paper.  The formulation chosen here sacrifices context for brevity, see \Cref{S:Rep1} for an extended discussion.  Let $\Lambda$ be a finite subset in $ \mathbb{Z}^{d}$ and for each $j\in \Lambda $,  $t_{j}$ is
a real variable. The $\Lambda$-tuple of numbers $\ut:= (t_j)_{j\in \Lambda}$ will be referred to as a (spin)-configuration and the collection of all such tuples $\BR^{\Lambda}$ will be referred to as the sample (or configuration)  space.  For any two points 
$j_1,j_2\in \Lambda $,  $|j_1-j_2|$ will denote the Euclidian distance
on the lattice.

On $\BR^{\Lambda}$, we define two functions $F$ and 
$M$ defined by the parameters $J_{jj'}, \eps_j\geq 0$:
\begin{align}
F_{\Lambda,J } (\nabla t)  &=\sum_{(jj')\in \Lambda }  
J_{jj'}(\cosh(t_j - t_{j'})-1), \cr
M^{\varepsilon }_{\Lambda } (t) &=  
\sum_{j\in \Lambda }  \varepsilon_{j} (\cosh t_j -1), 
\label{eq:FG}
\end{align}
where we denoted by $(jj')$ the nearest neighbor pairs $|j-j'|=1$. Next we define the  operator
 $D^{\varepsilon }_{\Lambda } (t)$ by
\begin{equation}\label{eq:Mmatrix}
\begin{array}{lll}
(D^{\varepsilon }_{\Lambda })_{jj'} & = 0  & |j'-j|>1\\
(D^{\varepsilon }_{\Lambda })_{jj'}   & = -J_{ij}   & |j'-j|=1\\
(D^{\varepsilon }_{\Lambda })_{jj}  & = +[\sum_{j'\sim j} J_{jj'}+V_{j}]
+ \varepsilon_{j} e^{-t_{j}} & j'=j\ ,\\
\end{array}
\end{equation}
where
\begin{equation}\label{eq:Vdef}
V_{j} =  \sum_{j', (jj')}  J_{jj'}\left[e^{t_{j'}-t_{j}} -1 \right].
\end{equation}
It is worth noting that one can rewrite this operator as $ e^{-t} \circ \Delta_{J, \eps}(t) \circ e^{-t}$  where $e^{-t}$ is the diagonal operator in the position basis $(e^{-t})_{ij}=\delta_{ij} e^{-t_j}$ and where $\Delta_{J, \eps}(t)$ is the sum of the weighted graph Laplacian with conductances $J_{jj'}e^{t_j+t_{j'}}$ and mass term (killing rates) $\varepsilon_{j} e^{t_{j}} $
We will mostly consider the case $\varepsilon_0=1$ and $\varepsilon_j=0$ otherwise, but the general setup is also of interest.  

Given these objects, we study the following Gibbs measures:
\beq
\label{E:t-meas}
 d\mu_{\Lambda, a, J, \varepsilon} (t) = Z_{\Lambda, a, J, \varepsilon}^{-1} \ \prod_{j\in \Lambda} \frac{dt_{j}}{[2\pi]^{1/2}} 
e^{-F_{\Lambda } (\nabla t)}e^{- M^{\varepsilon }_{\Lambda } (t)} 
\times 
\left[\det D^{\varepsilon }_{\Lambda } (t)  \right]^{a},
\eeq
where $dt_{j}$ is the Lebesgue measure on $\R$ and $a\in \R$ is a fixed parameter .  The partition function $Z_{\Lambda, a, J, \varepsilon} $ is defined to  normalize the integral to be a probability measure.  Note that if $\eps_j\equiv 0$ then  the model is not defined since a) $\det D^{\varepsilon }_{\Lambda } (t) $ vanishes identically and b) there is the non-compact   symmetry $t_j \mapsto t_j + x$.  Thus taking $\varepsilon_j \neq 0$ for some $j$ seems necessary for this measure to be normalizable.  Later, we will see that, nevertheless, one can consider the $\eps \rightarrow 0$ limit of the normalized measures (even before the thermodynamic limit).
For future convenience $\left \langle \cdot \right \rangle_{\Lambda, a, J, \varepsilon}=\left \langle \cdot \right \rangle_{J, \epsilon}$ denotes the corresponding Gibbs state (with $a, \Lambda$ fixed).  Our main interest is in the asymptotic behavior (with respect to $|k-\ell|$) of the  correlation functions 
\begin{align}
&\left \langle e^{z [t_k-t_{\ell}]}\right \rangle_{\Lambda, a, J, \varepsilon} \text{ $z\in \BC$},\\
&\left \langle G_{\ell k} \right \rangle_{\Lambda, a, J, \varepsilon} \text{ where $G= D^{-1}$},
\label{E:Gcor}
\end{align}
as the volume $\Lambda\uparrow \BZ^d$ .

This model is (the $t$-marginal distribution of) the titular \textit{Hyperbolic $\sigma$-model}.  The reason for this name is that when  $a\in \BZ\cup \BZ+1/2$, the measure $d\mu_{\Lambda, a, J, \varepsilon} (t)$ is a marginal distribution of the natural Gibbs measure on the space of maps from $\Lambda$ into some Hyperbolic space/superspace (see \Cref{S:Rep1}.  
Of most interest to us is the case when the coupling constants $J_{jj'}=\beta I_{jj'}$, where $I$ is the incidence matrix for $\Lambda$, but we have good reasons to keep the formulation general.   In this case $\beta >0$ is a parameter that can be interpreted
as an inverse temperature or, in the specific case $a=1/2$ as a measure of the amount of disorder in a system \cite{DSZ}.    

\subsection{Overview and Motivation}
Let us now discuss what is known and expected regarding the behavior of these models.  One may view them as (non-local, possibly non-convex) $\nabla \phi$ models.  From this point of view, if $a=0$ the  Gibbs weight reduces to the more familiar form $e^{-F_{\Lambda } (\nabla t)}e^{- M^{\varepsilon }_{\Lambda } (t)}$.  In this case the Gibbs measure is local and  log-concave and one expects that on large length scales the $t$-field behaves as a massless free field.  

If $a\neq 0$ the $\det D$ term brings an interesting nonlocal modification into the definition of the Gibbs measure.  When $a< 0$ the function $t \mapsto e^{-F_{\Lambda } (\nabla t)}e^{- M^{\varepsilon }_{\Lambda } (t)} 
\times 
\left[\det D^{\varepsilon }_{\Lambda } (t)  \right]^{a}$ is log-concave and we expect massless free field behavior to persist.  In particular off-the-shelf techniques, principally the Brascamp-Lieb inequality and its relatives,  give upper bounds on fluctuations of the $t$ field of this type.  This is not the end of the story however, since given a choice of pinning $\varepsilon_j$, it is important to understand the behavior of mean values $ \langle t_k-t_{\ell} \rangle_{\Lambda, a, J, \varepsilon} $ and to provide \textit{lower bounds} on fluctuations, especially when the points $k, \ell$ are far from the vertices being pinned.  To our knowledge this has been carried out fully only in the case $a=0$ by R.~Bauerschmidt (personal communication, September 2018), however the techniques of  \cite{SZ} may also be useful.

When $a>0$, the factor $\left[\det D^{\varepsilon }_{\Lambda } (t)  \right]^{a}$ is log-convex and competes with the more conventional term $e^{-F_{\Lambda } (\nabla t)}e^{- M^{\varepsilon }_{\Lambda } (t)}$.  This raises the possibility of a phase transition as e.g. $\beta$ varies (in the case $J_{jj'}=\beta I_{jj'}$).  
Indeed, in the specific case $a=1/2$ this does indeed occur when $d\geq 3$, \cite{DSZ,DS}, but this is the only case which has been looked at in detail.

Let us highlight some of their results. In \cite{DS}, the authors show that on any (infinite) graph with uniformly bounded degree, there is $\beta_0$ depending only on the degree of the graph so that $\langle G_{\ell k}\rangle_{J, \epsilon}$
decays exponentially fast provided $J_{jj'}\leq \beta_0$ and under mild conditions on the $\varepsilon_j's$.  On the other hand, in \cite{DSZ}, the authors show that on the graph $\BZ^3$ and provided $J_{jj'}=\beta I_{jj'}$, there is $\beta_0$ such that for all $\beta>\beta_0$, and if $\varepsilon_j=h$ independent of $j$, then for all $x, y\in \Lambda$
\[
\left \langle \cosh(t_x-t_y)\right \rangle_{J, \epsilon} , \left \langle \cosh(t_x)\right \rangle_{J, \epsilon} \leq 3
\]
independent of $\Lambda, h$ (technically, here the proof requires periodic boundary conditions). This statement does not quite address, but does support  the conjecture that if $\beta$ is large, the asymptotic behavior of the walk induce by the (random) operator $\Delta_{J, \eps}(t)$ is that of Brownian motion, (whereas the previous exponential decay result implies positive recurrence).
Thus the two results combine to show that on $\BZ^3$, the model has a phase transition as $\beta$ varies.

Besides being a beautiful result in its own right, the case $a=1/2$ connects to \textit{Edge Reinforced Random Walk} (ERRW) and to the \textit{Vertex Reinforced Jump Process} (VRJP), probabilistic models introduced and studied in a different context \cite{CD, BV}. Indeed, a remarkable paper by Sabot and Tarr\'{e}s \cite{ST} shows that the Gibbs measures introduced above are mixing measures for VRJPs.  They then used this connection, along with the ideas in \cite{DS} to prove positive recurrence under strong reinforcement.  We gave a less direct and computation proof of positive recurrence was  simultaneously given in \cite{ACK}.  We also note the excellent earlier work  \cite{merkl2009}, which proves recurrence on dilute versions of $\Z^2$ (which also makes use of harmonic deformation techniques).

Since \cite{ST, ACK},  a number of related papers appeared, in particular \cite{STZ,SZ1, S, BHS} which demonstrate recurrence/recurrent-like signatures for ERRW and VRJP on $\BZ^2$.  In particular, the present work was inspired by the insightful paper \cite{BHS}.  The main observation of that paper can be loosely summarized as saying that the generator of the VRJP coincides with the adjoint, with respect to the Gibbs measure, of the differential operator generating  global hyperbolic symmetry.  In principle their work suggests that the VRJP is relevant for \textit{all} Hyperbolic$\sigma$-models, not just when $a=1/2$.  While this turns out to be true for $a\leq 0$ (in this case $ d\mu_{\Lambda, a, J, \varepsilon} (t)$ provides a nontrivial stationary mixing measure for the VRJP),  the caveat for $a\in \BN+1/2$ is that the coefficients of the generator for the VRJP become Grassmann-valued (R. Bauerschmidt, personal communication, December 2018).  It remains to be seen whether probabilistic sense can be made of this  last remark.

A major open problem remaining after \cite{DS,DSZ}is the following:  Conventional wisdom, based on physical reasoning due to Polyakov (see Part 4 of the lecture notes \cite{Tong} for a nice account of this) suggests that on $\BZ^2$, $\langle G_{\ell k} \rangle_{J, \epsilon}$, with, say, $J_{jj'}=\beta I_{jj'}$ and $\varepsilon_0=1,  \varepsilon_j=0$ otherwise, decays exponentially in $|\ell-k|$. While the rate of decay should depend rather nontrivially on $\beta$, the fact that there is exponential decay should not.  This is a special case of the longstanding open problem to demonstrate mass generation for two dimensional nonlinear $\sigma$-models (whose $\beta$ functions grow in the infrared).  
Thus our initial interest: proving this result for $a=1/2$ seemed out of reach, but wehoped that, by considering other values of $a$, in particular $a=3/2, 5/2, 7/2, \dotsc$ and an alternative description to be explained below, one might make progress (the physics calculations get 'better' as $a$ increases.  These hopes now seem naive.  Nevertheless we compile evidence below that the study of the model at odd half integers is, in itself, an incredibly interesting enterprise.

\subsection{Summary of Results and Discussion}
Let us now detail our findings.  Our modest initial goal was to investigate whether, with $a>0$, one could apply harmonic deformation techniques (e.g. the Mermin-Wagner theorem or McBryan-Spencer method) to say anything about the behavior of $e^{[t_k-t_{\ell}]}$ in large volumes $\Lambda$ and for $k, \ell$ far apart from one another.  There are a few reasons to start with this question.  For one, as already mentioned, the problem of mass generation in $d=2$ remains a goal.  While harmonic deformation techniques have no hope of addressing the question directly, they allow us to familiarize ourselves with the objects of study while also providing a good chance that new results may be derived.  

A second reason is that, in case $a=1/2$, these techniques were recently  applied independently by G.~Kozma and R.~Peled \cite{KP} and \cite{S} to obtain a crucial \textit{a priori} estimate for recurrence of the VRJP in $2d$.  Kozma and Peled apply the harmonic deformation technique configuration-wise.  In order for this approach to succeed, they require an \textit{a priori} statement to the effect that the locations where $\cosh(t_j-t_j')$ is large, for nearest neighbors $(jj')$, are very sparse.  To show this latter fact, they use the dual description of  $d\mu_{\Lambda, a=1/2, J, \varepsilon} (t) $ as the \textit{mixing measure} for the VRJP.  On the other hand, Sabot avoids working with the VRJP by relying instead on SUSY through the fact that the partition function $Z_{\Lambda,  a=1/2, J, \eps} \equiv 1$ for all $J_{jj'}$ (the reason this is important will become clear below).

Let us now state our first result.
\begin{theorem}[Bounds of Fourier Transforms]
\label{T:Main}
Fix $a\in \BN+1/2$.  Let $J_{jj'}=\beta I_{jj'}$ and choose $\varepsilon_0=1, \varepsilon_j=0$ otherwise.  For all $\Lambda$ sufficiently large, we have the bounds
\begin{align*}
&\left|\left \langle e^{ik(t_m-t_\ell)} \right \rangle_{\Lambda, a, J, \varepsilon_0} \right|\leq \exp\left(-\frac{k^2}{\beta+2a}\left[\log(1+\|m-\ell\|_{2})-\log\left(1 +\frac{4|k|}{(\beta+2a)}\right)\right]\right),\\
&\left|\left \langle e^{ik(t_m-t_\ell)} \right \rangle_{\Lambda, a, J, \varepsilon_0} \right|\leq \exp\left(-\frac{[|k|-\beta-a]}{2}\left[\log(1+\|m-\ell\|_{2})\right]\right).
\end{align*}
\end{theorem}
This gives non-concentration of the $t$-field analogous to a massless Gaussian free field.   

Next by combining an inequality appearing in \cite{S} with the main ingredient in our proof of Theorem \ref{T:Main} (which we will broadly outline in the subsequent two paragraphs) we can bound Laplace transforms for models with $a\in \N+1/2$ (in case $a=3/2$ this bound appears in \cite{H02} as well).
\begin{theorem}[Bounds on Laplace Transforms]
Let $a\in \N+1/2$, and $0<p<2a$ be fixed.
\label{T:Laplace}
  On $\BZ^2$ with $\beta_{ij} = \beta 1_{i\sim j}$ there is $c(\beta,p) > 0$ such that
  \begin{equation}
\left \langle e^{p(t_v-t_0)} \right \rangle_{\Lambda, a, J, \varepsilon_0}   \leq |\dist(v, 0)|^{-c(\beta, p)}.
  \end{equation}
\end{theorem}

The reason for the restriction $a\in \BN+1/2$ should be viewed through an analogy between Potts and random cluster models.  While the Random Cluster Models make sense for any value of $q>0$, it is only for integer values of $q$ that is one able to rewrite the model as a spin system.  Having both descriptions allows powerful tools to be applied in the integer case which are otherwise unavailable, e.g. Reflection Positivity.  So too, as we (partially) explain in \Cref{A:SUSY} and use from \Cref{S:Rep1} onwards, when $a\in\BZ \cup [\BZ+1/2]$ the measure $d\mu_{\Lambda, a, J, \varepsilon} (t)$ maybe realized as the marginal distribution, in horospherical coordinates, of a spin model over $\Lambda$ in which the spins take values in,  respectively, a hyperbolic space $\BH^{-2a}$ if $a\leq 0$ or in a hyperbolic superspace if $a>0$ (either $\BH^{1|2N}$ if $2a$ is an even integer or $\BH^{2|2a+1}$ if $2a$ is an odd integer). 

There are two ingredients that we make use of in proving each these theorems.  In each case, the first is an \textit{a priori} estimate.  In the case of the  Fourier transform, we use a technique with roots in the classic paper of McBryan and Spencer \cite{MS}.  It amounts to bounding the Fourier transform of the distribution for the variable $t_{k}-t_{\ell}$ by shifting contours.  For hyperbolic $\sigma$-models, more so than their compact counterparts (the subject of the original paper \cite{MS}), this step is essentially the method used to compute the characteristic function of a Gaussian variable.  This technique and an associated bound on the Fourier transform of the distribution for $t_k-t_\ell$ is presented in the following section \Cref{S:MS}.  

In the case of the Laplace transform, for this first step we instead apply a slight generalization of a recent bound of C. Sabot \cite{S}.   We state this estimate at the beginning of Section \ref{S:Laplace} and refer the reader to \cite{S} for the original poof when $a=1/2$, or to \cite{H02} for the general case.  In both cases, the  upshot  will be that to derive both \Cref{T:Main,T:Laplace}, we need to show that if we consider the partition function $Z_{\Lambda, a, J, \varepsilon}$ as a function of the single coupling $J_{jj'}$ for some fixed $j, j'$ holding the remaining variables fixed, $F(J_{jj'}):= Z_{\Lambda, a, J, \varepsilon}$, then $F$ is necessarily increasing in $J_{jj'}$.   

Thus, the heart of our paper is the following remarkable collection of  facts. We state them in the general context of hyperbolic SUSY $\sigma$-models on a finite weighted graph $(G, J)$.  The notation should be self-evident.
\begin{theorem}
\label{T:Main2}
Fix $a\in \BN+1/2$, let $G=(V, E)$ be a finite graph and choose $\varepsilon, J$ to be non-negative masses and couplings.  Let us fix a nearest neighbor pair $j, j'\in E$ and view
\[
J_{jj'}\mapsto Z_{G, a, J, \varepsilon}
\]
as a function from $\BR_+$ to $\BR_+$.   Then 
\begin{itemize}
\item $Z_{G, a, J, \varepsilon}$ is a polynomial of degree $n=a-1/2$ in $J_{jj'}$.

\item
For all $k\leq \lceil n/2\rceil $ and for $k=n$, $\partial_{J_{{jj'}}}^k Z_{G, a, J, \varepsilon}>0$. 

\item Specializing $\eps$ to be supported at one vertex $v$ and $0\leq \eps_v \leq 1$ then $\partial_{J_{{jj'}}}^k Z_{\Lambda, a, J, \varepsilon_v}\geq 0$ for all $k\leq n$.  
\end{itemize}
In particular $Z_{\Lambda, a, J, \varepsilon}$ is increasing in each $J_{jj'}$.
\end{theorem}
We strongly suspect that, in the third conclusion, the restriction to sufficiently small one point pinning is an artifact of our proof, and we would very much like to remove this hypothesis.

Recall that  $\partial^k_{J_{jj'}} Z_{\Lambda, a, J, \varepsilon}$ are proportional to the moments of the interaction between spins at $j, j'$.  We will prove non-negativity of polynomial coefficients by bounding these moments. In particular, 
$Z_{\Lambda, a, J, \varepsilon}$ is increasing if the first moment of the interaction is positive  for any positive choice of $J$'s.
Positivity of this first moment turns out to be substantially easier to demonstrate than positivity in the general case.  We have opted to go the extra mile and prove positivity of all moments (up to $a-1/2$) for reasons  to be discussed after Theorem \ref{T:Sokal-etal}.

Let discuss the context of this result with respect to  the previously mentioned special case, when $a=1/2$.  In that case the model has the fundamental property that $Z_{J, \varepsilon}=1$ for any choice of non-negative $J, \varepsilon$.  This and other basic identities play a crucial role in \cite{DS, DSZ, S}.  As explained in an appendix of \cite{DSZ}, this fundamental identity follows from what physicists refer to as 'Supersymmetric Localization'.  The same technique goes by the name the Duistermaat-Heckman Theorem \cite{DH} in the mathematics community.  For us, the consequence of this technology is that, after disintegrating $d\mu_{\Lambda, a, J, \varepsilon} (t)$ into the full $\H^{2|2N}$ spin model,
we may rewrite the system as the 'Gibbs state' of a $2(N-1)$ component pure Grassmann field if $a\in \BN+1/2$.  That $Z_{\Lambda, a, J, \varepsilon}$ is a polynomial in $J_{jj'}$ is then manifest.  Where we have to work is to show the second statement, positivity of the polynomial's coefficients.

There are a few hints that such a positivity might be true.  For one it is possible to compute by hand what is going on when $a=3/2,5/2$.  In these cases we initially managed to show the positivity provided the $J_{jj}$'s are uniformly large.  However, what truly convinced us that it must be true is the following theorem, which combines the observation that localization reduces the $\H^{2|2N}$ model to a purely Grassmann variable field theory, along with a beautiful algebraic paper \cite{Sokal-etal} which coincidentally connects the latter field theory with \textit{unrooted} spanning forests (the reader may also consult the forthcoming \cite{H02} for an abbreviated account of this development).
To state the result, we need to introduce this last object.  Given a finite graph $\CG$, let $ \CF(G)$ denote the collection of \textit{unrooted} spanning forest on $G$ and, for $F\in \CF(G)$, let 
\[
W(F)= \Biggl( \prod\limits_{ij \in F} J_{ij} \! \Biggr)
      \; \prod_{T \in F}[1+\sum_{i\in T}\varepsilon_i]
\]
\begin{theorem}
\label{T:Sokal-etal}
Let $a=3/2$ and let $G$ be fixed and finite.  For any $\eps_i,  J_{ij}\in \BR$,
\beq
\label{Eq:USF1}
Z_{G, a, J, \varepsilon}=2 \sum_{F\in \CF(G)} {W(F)}
\eeq
where $\prod_{T\in f}$ denotes the product over components (trees) of $F$.
\end{theorem}
As the reader may anticipate, this correspondence extends also to many types of correlation functions, for example  $\langle G_{\ell k} \rangle_{G, a, J, \varepsilon}$ is equal to the probability $\ell$ and $k$ are connected in the probability measure determined by the weights $W(F)$.  Thus, our hyperbolic nonlinear $\sigma$-model with $a$ set to $3/2$ provides a continuous representation of a natural class of probability measures on unrooted spanning forests.  Together with Bauerschmidt, Helmuth and Swan, we explore this connection in some detail \cite{H02}, proving in particular that in two dimensions, there are no infinite trees in a thermodynamic limit for any $\beta$ (the thermodynamic limit should be unique, but we do not have the technology to prove that).  

Theorem \ref{T:Sokal-etal} raises the question as to whether, given that the models with $a=1/2$ and $a=3/2$ have 'dual representations'  in the discrete probability world, there are such representations for all $a\in \BN+1/2$.  This question is wide open.  We regard the positivity expressed in \Cref{T:Main2} as an important contribution above and beyond the fact that the partition functions increase coordinate-wise as it provides a consistency check: If such a correspondence did exist, \Cref{T:Main2} would follow for free.

The plan of the remainder of the paper is as follows.  Postponing the introduction of Grassmann variables as long as possible, in the next section we present a bound on the characteristic function of the random variable $t_k-t_{\ell}$.  This computation holds for any $a\in \BR$ and brings us to the fundamental problem of controlling partition function ratios $\frac{Z_{J'}}{Z_{J}}$ when coordinate-wise we have $J'_{jj'}\leq J_{jj'}$.  In \Cref{S:Rep1} we finally reveal the Grassmann variable representation of the hyperbolic $\sigma$-models of interest, summarizing just what we need to continue on to prove \Cref{T:Main2}. For convenience we provide a more detailed  exposition in  \Cref{A:SUSY} and we also refer the reader to the Appendices of \cite{DSZ} on which our discussion is based.   In \Cref{S:Ward}  we develop some important identities available in the Grassmann representation.  For a bit of background on the origin of these identities, the reader may consult \Cref{S:R12n} and Section 7 of \cite{Sokal-etal}   These identities form the basis for the computations which follow in \Cref{S:Pos} that demonstrate positivity of the polynomial coefficients.  In Sections \ref{S:Fourier} and \ref{S:Laplace} we then complete the proofs of \Cref{T:Main,T:Laplace}.
 
 \subsection*{Acknowledgements}
I thank Roland Bauerschmidt for explaining to me the results in \cite{BHS}, which motivated the present work. I also thank Bauerschmidt, Tyler Helmuth and Andrew Swan for numerous discussions on related topics throughout the preparation of this manuscript. 

\section{Upper Bounds  by Partition Function Ratios}
\label{S:MS}
In this section $d=2$ and $\Lambda \subset \BZ^2$.  Abusing notation slightly, we let  $\eps_0$ denote the pinning vector $\varepsilon_0=1$ and $\varepsilon_i=0$ in subscripts where a pinning vector may appear, e.g. $Z_{\Lambda, a J, \eps_0}$. Let $a\in \BR$ and $J_{jj'}>0$ be fixed.  
Following McBryan and Spencer, we consider the effect of translating the integration variables $t_x$ into the complex plane, $t_x\mapsto t_x+\textrm{i}\rho_x$. 
Given $\rho:\Lambda \rightarrow \BR$, and a collection of coupling constants $J_{jj'}$,
let $J_{jj'}(\rho)=J_{jj'} \cos(\rho_j-\rho_j')$.  
Also, let 
\[
\CA:=\{ \rho: \text{ $\rho_0=0$ and $\cos(\rho_j-\rho_{j'})>0$ for all $(jj')$}\}.
\] 
We begin with a general estimate:
\begin{lemma}
\label{L:FT}
If $\rho\in \CA$,
\[
\left|\left \langle e^{\textrm{i} k[t_m-t_\ell]}\right \rangle_{\Lambda, a, J, \eps_0}\right|\leq \frac{Z_{\Lambda, a, J(\rho), \eps_0}}{Z_{\Lambda, a, J, \eps_0}} \prod_{(jj')} \cos(\rho_j-\rho_j')^{{-\max(a,0)}} e^{ \sum_{(jj')\in \Lambda } J_{jj'} [1-\cos(\rho_j-\rho_j')]}  e^{-k[\rho_m-\rho_\ell]}.
\]
\end{lemma}
We prove this estimate below.  Before launching into the proof, however, let us record, for convenient comparison, an analogue of Lemma \ref{L:FT} for Laplace transforms.  The proof of this latter bound appears in \cite{S} in the special case $a=1/2$, but holds \textit{mutatis mutandis} for general $a$, see \cite{H02} for the derivation modulo an   \textit{a priori} moment bound which is of independent interest which we state and prove later \Cref{L:Moment}.  Except for this moment bound,  the proof of \Cref{L:Laplace} is omitted from the present note.

\begin{lemma}
\label{L:Laplace}
Fix $0<s< 1$.  Let $\rho: \Lambda\rightarrow  \R$ be given so that $\rho_0=0$ and $\rho_v=1$.  Choose $q>1$ so that $s+1/q=1$ and choose $\gamma>0$ so that $q^2\gamma \|\nabla \rho\|_{\infty}\leq 1/2$.  Then
\[
\langle e^{s 2a t_v} \rangle_{\Lambda, \beta, a, \eps_0}   \leq  e^{-2as \gamma} e^{ \beta  \sum_{(jj')}  q^2\gamma^2 (\rho_j-\rho_{j'})^2}
      \frac{Z_{\tilde \beta}}{Z_\beta}\prod_{(jj')\in E}\left( {\frac{\beta}{\tilde{\beta}_{jj'}}}\right)^{a\vee 0}
      \]
  where
  \begin{equation}
    \tilde \beta_{jj'} = \beta(1-2q^3\gamma^2(\rho_j-\rho_j')^2).
  \end{equation}
\end{lemma} 

\begin{proof}[Proof of \Cref{L:FT}]
Let us
consider the integral
\[
I(\rho):= \int \prod_{j\in \Lambda}  \frac{dt_{j}}{\sqrt{2\pi }} \: e^{\textrm{i} k[t_m-t_\ell]-k[\rho_m-\rho_{\ell}]}\
e^{- \sum_{(jj')\in \Lambda } 
J_{jj'}[\cosh(t_j - t_{j'} + \textrm{i}[\rho_j-\rho_{j'}])-1]}e^{- M^{\varepsilon }_{\Lambda } (t)} [\det D^{\varepsilon_0}_{\Lambda, J } (t+i \rho)]^a.
\]
We begin by arguing that the integrand is in $L^1$ if $\rho_0=0$ and $\cos(\rho_j-\rho_{j'})>0$ for all $(jj')$, that is if $\rho \in \CA$.  Under this assumption, let 
$\alpha=\min_{(jj')}\cos(\rho_j-\rho_{j'})>0$.

For general complex-valued $t$, let us denote its real and imaginary parts by $\Re(t), \Im(t)$ respectively.  Then we may denote $J_{jj'}(\Im(t))=J_{jj'}\cos(\Im(t_j-t_{j'}))$ and we observe that
\[
\left | e^{-F_{\Lambda, J} (t)}  \right|\leq C(\Lambda) e^{-F_{\Lambda, J(\Im(t)) } (\Re( t))}.
\]
By the assumption on $\rho$, $J_{jj'}(\Im(t))\geq J_{jj'}\alpha>0$,  allows us to control the RHS.
Also, the Matrix Tree Theorem implies  the immediate point-wise bound
\[
|\det D^{\varepsilon }_{\Lambda, J  } (t) |\leq \left[\det D^{\varepsilon }_{\Lambda, J } (\Re(t)) \right] 
\]
Combining these estimates implies that for $\rho \in \CA$, the  integrand appearing in the definition of $I(\rho)$ is in $L^1$ and that
\[
|I(\rho)|\leq   C(\Lambda) e^{k[\rho_m-\rho_{\ell}]} \int \prod_{j\in \Lambda}  \frac{dt_{j}}{\sqrt{2\pi }} \:
e^{- \sum_{(jj')\in \Lambda } 
J_{jj'} \alpha [\cosh(t_j - t_{j'})-1]}e^{- M^{\varepsilon }_{\Lambda } (t)} 
[\det D^{\varepsilon_0}_{\Lambda, J } (t)]^a< \infty.
\]

Next we wish to argue that $I(\rho)$ is locally constant as a function on $\CA$, and hence constant over all of $\CA$ by connectedness.
With $\alpha=\min_{(jj')}\cos(\rho_j-\rho_{j'})>0$, we claim that there is $\epsilon=\epsilon(\alpha)>0$ such that if  $\psi$ is chosen with $\psi_0=0$ and $\|\psi\|_\infty \leq \epsilon$ then $I(\rho)=I(\rho+\psi)$.
To see this, consider the collection of contours (in $\C$)
\begin{multline}
\label{eq:contour}
\Gamma_K(j)=
[-K+\textrm{i}\rho_j, K+\textrm{i}\rho_j]\circ[K+\textrm{i}\rho_j, K+i(\rho_j+\psi_j)]\\\circ [-K+i(\rho_j+\psi_j), K+i(\rho_j+\psi_j)]\circ [-K+i\rho_j, -K+i(\rho_j+\psi_j)].
\end{multline}
By Green's theorem,
\beq
\label{eq:shift}
\oint_{\times_{j\in \Lambda} \Gamma_K(j)} \prod \textrm{d} t_j   e^{\textrm{i} k[t_m-t_\ell]} e^{-F_{\Lambda, J} (\nabla t)}e^{- M^{\varepsilon }_{\Lambda } (t)} 
\times 
\left[\det D^{\varepsilon }_{\Lambda, J  } (t)  \right]^{a}=0,
\eeq
where we used that $a>0$ so that no contribution is picked up from $0$'s of $\det D^{\varepsilon }_{\Lambda, J  } (t)$.

We next show that the contribution to this multivariate contour integral  from $\{\ut: |\Re(t_y)|=K \text{ for some $y$}\}$, i.e. from the second and fourth contours in \eqref{eq:contour}, is negligible for $K$ large. First, 
\begin{multline}
\label{E:1}
\left
|\oint_{\times_{j\neq y} \Gamma_K(j)} \prod \textrm{d} t_j  \: e^{\textrm{i} k[t_m-t_\ell]} e^{-F_{\Lambda, J} (\nabla t)}e^{- M^{\varepsilon }_{\Lambda } (t)} 
\times 
\left[\det D^{\varepsilon }_{\Lambda, J  } (t)  \right]^{a}\right|\\
\leq  C(\Lambda) \int_{\times_{j\neq y} \Gamma_K(j)} \prod |\textrm{d} t_j  |  \: e^{-F_{\Lambda, J\alpha} (\Re(\nabla t))- M^{\varepsilon }_{\Lambda } (\Re(t))} 
\times 
\left[\det D^{\varepsilon }_{\Lambda, J} (\Re(t)) \right]^a.
\end{multline}
If $y$ is fixed with  $|\Re(t_y)|=K$ then by either $|t_y-t_0|\geq K/2$ or $|t_0|\geq K/2$. In the first case the must be some edge $(jj')$ in $\Lambda$ such that $|t_j-t_j'|\geq K/2|\Lambda|$.  Therefore
\[
e^{-F_{\Lambda, J/\sqrt{2}} (\Re(\nabla t))- M^{\varepsilon }_{\Lambda } (\Re(t))}\leq e^{-\min_{jj'} J_{jj'}\cdot \alpha\cdot [\cosh(K/2|\Lambda|)-1]\wedge\eps_0 [\cosh(K/2)-1]}
\]
due to the pinning at $0$.  Via the matrix tree theorem again
\[
|\left[\det D^{\varepsilon }_{\Lambda, J} (\Re(t)) \right]|\leq C(\Lambda)e^{|\Lambda| K},
\]
so that the RHS of \eqref{E:1} is further bounded by
\[
C_1(\Lambda)e^{|\Lambda| K} [K\|\rho\|_{\infty}]^{2|\Lambda|}  e^{-\min_{jj'} J_{jj'}\cdot \alpha\cdot [\cosh(K/2|\Lambda|)-1]\wedge\eps_0 [\cosh(K/2)-1]}.
\]
Thus the contributions to the contour integral from $\{\ut: |\Re(t_y)|=K \text{ for some $y$}\}$ tend to $0$ as $K$ tends to $\infty$.

Thus applying the dominated convergence theorem to the remaining contour of the identity \eqref{eq:shift},  we see that we can shift each $t_j$ integral from $\BR+i\rho_j$ to $\BR+i(\rho_j+ \psi_j)$, provided $\|\psi \|_{\infty} \leq \epsilon$ and $\rho \in A$. We obtain
\[
I(\rho)=I(\rho+\psi) 
\]
in this case. Finally, connectedness then imnplies that $I(\rho)=I(0)$ throughout $A$.  

To finish the lemma we need to estimate $I(\rho)$ more precisely.  
We have
\begin{multline}
\left|\left \langle e^{\textrm{i} k[t_m-t_\ell]}\right \rangle_{\Lambda, a, J, \eps_0}\right|= \frac{|I(\rho)|}{Z_{\Lambda, a, J, \eps_0}} \leq e^{-k[\rho_m-\rho_{\ell}]- \sum_{(jj')\in \Lambda } 
J_{jj'}[\cos(\rho_j - \rho_{j'})-1]} \\
\times \underbrace{\frac{1}{Z_{\Lambda, a, J, \eps_0}} \int \prod_{j\in \Lambda}  \frac{dt_{j}}{\sqrt{2\pi }} 
e^{- \sum_{(jj')\in \Lambda } 
[J_{jj'}(\rho)[\cosh(t_j - t_{j'})-1] }e^{- M^{\varepsilon }_{\Lambda } (t)} [\det D^{\varepsilon_0}_{\Lambda, J } (t)]^a}_{\textrm{I}}.
\end{multline}
We observe that we can rewrite the integral on the RHS as
\beq
\textrm{I}=\frac{Z_{J(\rho),\eps_0}}{Z_{J, \eps_0}} \Bigg \langle\left(\frac{\det D^{\eps_0}_{\Lambda, J} (t)}{\det D^{\eps_0}_{\Lambda, J(\rho)} (t)}\right)^a \Bigg \rangle_{\Lambda, a, J(\rho), \eps_0}
\eeq
Using the Matrix Tree Theorem again
we have the point-wise bound (valid for any $a$)
\[
\left(\frac{\det D^{\varepsilon}_{\Lambda, J} (t)}{\det {D}^{\varepsilon}_{\Lambda, J(\rho) } (t)}\right)^a\leq \prod_{(jj')\in \Lambda} \cos(\rho_j-\rho_j')^{-\max(a,0)}
\]
and this finishes the proof.
\end{proof}

\section{A Brief Introduction to the Dual Grassmann Representation  for $a\in \BN+1/2$.}
\label{S:Rep1}
In the next four sections, we work on a general finite connected weighted graph $G=(V,E, J)$ where the edge-weights  $J_{jj'}>0$ for $(j, j')\in E$.  In defining the pinning vector $\eps_0$, choose some distinguished vertex $0\in V$ to be viewed as the origin. We extend the definition of $d\mu_{G, a, J, \varepsilon} (t)$ and the associated Gibbs state in the obvious way.
As remarked in Lemma \ref{L:FT} for the Fourier Transform, or from an analogous estimate for the Laplace transform in Section \ref{S:Laplace}, beyond the classical harmonic deformation technique, we need to control $\frac{Z_{G, a, J', \eps_0}}{Z_{G, a, J, \eps_0}}$  assuming component-wise domination $J'_{jj'}\leq J_{jj'}$.  That this is possible is the main special feature of the models of interest which we observe in this paper.

In this section our goal is to provide the reader with a minimum of necessary notation and identities to proceed with the proof \Cref{T:Main2}.  More details may be found in \Cref{A:SUSY} or in Appendix C of \cite{DSZ}.  The result being quoted relies on supersymmetric localization.
Given $n$, let $(\psi_i^{\ell}, \psibar_i^{\ell})_{\ell=1, i\in \Lambda}^{n}$ be a system of generators of the Grassmann algebra $\CG_{V}$ with $2n$ variables per site.
Let 
\begin{eqnarray*}
\psibar_i\cdot \psi_i=\sum_{\ell=1}^{n} \psibar_i^{\ell}\psi_i^{\ell},& \quad \sigma_i=\sqrt{1+2\psibar_i\cdot \psi_i}, \quad &
D\mu_0(\psibar_i, \psi_i)=\prod_{\ell}\partial_{\psibar_i^{\ell}} \partial_{\psi_i^{\ell}}
    \circ \sigma_i^{-1},\\
    \CD_0(\psibar, \psi)= \prod_{i\in V, \ell} \partial_{\psibar_i^{\ell}} \partial_{\psi_i^{\ell}},&
  \CD(\psibar, \psi)=\prod_{i \in V} D\mu_{0}( \psibar_i, \psi_i), &\\
\end{eqnarray*}
and define the Grassmann action
\[
S_{J, \epsilon}=\sum_{(jj')} J_{jj'}\left\{\sigma_j \sigma_{j'} -\psibar_{ j}\cdot {\psi}_{j'}-\psibar_{ j'}\cdot {\psi}_{j}-1\right\}+\sum_i \varepsilon_i [\sigma_i-1].
\]

\noindent
\textbf{Superintegral sign convention:}  Note that our sign convention was chosen so that
\[
\int D\mu_0(\psibar_i, \psi_i) \cdot \sigma_i \exp(\lambda \psibar_i\cdot \psi_i)=(1- \lambda)^n
\]
so as to conform with the convention chosen in \cite{DSZ}.  This convention clashes with the choice in another key paper \cite{Sokal-etal}.  

\noindent
\textbf{Notation:}  We now introduce the partition function and Gibbs state for the Grassmann field.  We have
\begin{align}
&Z^f_{G, a, J, \eps}:=\int\CD(\psibar, \psi)e^{-S_{J, \eps}},\\
&\langle \cdot \rangle^f_{G, a, J, \eps}:=\frac{\int\CD(\psibar, \psi) \cdot e^{-S_{J, \eps}}}{{Z^f_{G, a, J, \eps}}}.
\label{eq:G}
\end{align}
In the rest of this paper, we refer to this model as that $\H^{0|2n}$ $\sigma$-model.

Note that, to distinguish the present model from the Gibbs states for the $t$-field or the full $\H^{2|2a+1}$-valued spins (discussed in \Cref{A:SUSY}), for expressions involving the purely Grassmann field we use the superscript $f$.  This convention will remain enforced throughout the rest of the paper.  We are now ready to state the main reason for introducing the Grassman model above.
\begin{lemma}
\label{L:SUSY1}
Let $\eps$ be a nonzero pinning vector (pinning at one or multiple points).  Let $a\in \BN\cup \{0\}+1/2$ and $2n=2a-1$.  Then we have a 'duality' between the $(t_i)_{i\in \Lambda}$ variables and  the $(\psi_i^{\ell}, \psibar_i^{\ell})_{\ell=1, i\in \Lambda}^{n}$.  At the level of partition functions, this duality takes the form
\beq
\label{E:SUSY2}
Z_{G, a, J, \eps} \textbf{=}
Z^f_{G, n, J, \eps}.
\eeq
When $n=0, a=1/2$ the identity reads
\beq
Z_{G, a, J, \eps} =1.
\eeq
\end{lemma}
\noindent
It is worth remarking that it is not obvious from its definition that $Z^f_{G, a, J \eps}>0$.  The previous lemma confirms this, and therefore demonstrates how the connection between the two sets of variables can be used in both directions.

If $a>1/2$ one obtains further identities for correlation functions.  We now state the most general one we will need in this paper.  Its proof is a straightforward exercise using \Cref{lem:SUSYloc} and is omitted.
We introduce some shorthand notations to be used here and in the remainder of the text.  Let
\begin{align}
\label{E:tau}&\tau_{ij}^\ell=- [\psibar_i^{\ell}- \psibar_j^{\ell}][\psi_i^{\ell}-\psi_j^{\ell}], \eta_{ij}^{\ell}=-[\psibar_i^{\ell}-\psibar_{j}^{\ell}]\psi_0^\ell,\\
& \pi_{i}^\ell=-\psibar_i^{\ell} \psi_i^{\ell}\text{ so that } \sigma_i=\sqrt{1-2\sum_{\ell = 1}^n \pi_i^{\ell}}.
\end{align}
Observe that
\beq
[\psibar_i^m-\psibar_j^m]^2= [\psi_i^m-\psi_j^m]^2 =[\pi_i^m]^2=0.
\eeq
Furthermore, recall the ($t$-dependent) Green kernel $G_{ij}$ defined at Equation \ref{E:Gcor} and let
\beq
G^{(1)}_{ij}:= [G_{ii}+G_{jj}-2G_{ij}], \quad G^{(2)}_{ij}:=G^{(1)}_{ij}G_{00}-[G_{i0}-G_{j0}]^2 \quad G^{(3)}_{ij}=[G_{i0}-G_{j0}].
\eeq
Note that $G^{(1)}_{ij}, G^{(2)}_{ij} \geq 0$ point-wise with respect to the $t$-field.

\begin{lemma}
\label{L:GIBP} 
Fix $a, n$ so that $2a-1=2n$ and so that $n\in \N$. Suppose $J, M \in \BN$ satisfy $J+M\leq n$.
For any $A\subset [J+M]$ and denoting $|A\cap [J]|=a, |A|=b$,
\beq
\label{E:Moment5}
\left\langle \prod_{\ell=1}^{J} \tau_{ij}^{\ell} \times   \prod_{j\in A} \pi_0^j \times \prod_{M+J+1}^{n} \eta_{ij}^{m} \right \rangle^f_{G, n, J, \eps_0}=\left\langle {G_{ij}^{(1)}}^{J-a} {G_{ij}^{(2)}}^{a} {G_{00}}^{b-a} {G^{(3)}_{ij}}^{n-(M+J)}\right \rangle_{G, a, J, \eps_0},
\eeq
where on the LHS spins reside in $\H(0|2n)$ while on the RHS spins reside in $\H(2|2(n+1))$.
\end{lemma}
Note in particular that if $n-(M+J)$ is even then this expression is non-negative.  This will always be the case below.

\section{Proof of Theorem \ref{T:Main2}}
\label{S:Pos}
From now on, we fix a graph $G$ and $n\in \N$ and suppress them from the notation.
To prove \Cref{T:Main2}, we pass back and forth between the two descriptions of the partition function/Gibbs state - in terms of the $t$ field on one hand and in terms of the pure Grassmann field on the other.  
We first encounter this idea in arguing that $Z_{G, a, J, \eps}$ must be a multivariate polynomial in \Cref{C:Poly} below.

\subsection{Upper Triangulating the Grassmann variables}
\label{S:UT}
We want to find a coordinate system of $2n$-component Grassmann variables in which 
\[
-\frac{1}{2} (\n_i-\n_j\,,\n_i-\n_j)= 1-\sigma_i \sigma_j +\bar{\psi}_{ i}\cdot \psi_{ j}+\bar{\psi}_{ j} \cdot \psi_{ i}
\]
can be expressed in terms of expressions like $[(\psibar_{j}^{\ell}-\psibar^{\ell}_{j'})( \psi^{\ell}_{j}     - \psi^{\ell}_{j'})  ] $ and $[\psibar^{\ell}_{j}\psi^{\ell}_{j}\psibar^{\ell}_{j'}\psi^{\ell}_{j'}]$.
If we have such a coordinate system then due to the fact that, for any single component $\ell$,
\[
[(\psibar_{j}^{\ell}-\psibar^{\ell}_{j'})( \psi^{\ell}_{j}     - \psi^{\ell}_{j'})  ] ^2        = [\psibar^{\ell}_{j}\psi^{\ell}_{j}\psibar^{\ell}_{j'}\psi^{\ell}_{j'}]^2=0,
\]
we can conclude 
\[
[\frac{1}{2} (\n_i-\n_j\,,\n_i-\n_j)]^{n+1}=0.
\]
Applied to $e^{-S_{J, \eps}}$, we then see that $Z_{G, a, J, \eps}$ must be a multivariate polynomial of degree at most $n$ in each $J_{jj'}$.

Let us introduce the notation
\[
\sigma_i(J)=\sqrt{1-2\sum_{\ell \leq J} \bar{\pi}_i^J}.
\]
We have
\[
 1-\sigma_i \sigma_j =1-\sigma_i{(n-1)}\sigma_j{(n-1)}  +\pi^{n}_i\frac{\sigma_j{(n-1)}}{\sigma_i{(n-1)}}+\pi^{n}_j\frac{\sigma_i{(n-1)}}{\sigma_j{(n-1)}} -\frac{\pi^{n}_i \pi^{n}_j}{\sigma_i{(n-1)} \sigma_j{(n-1)}}.
\]
We can now change variables, setting
\begin{align}
\sigma_i{(n-1)} \bar{\phi}_{i}^{n}:= \bar{\psi}_i^{n} &\quad   \sigma_i{(n-1)}{\phi}_i^{n}:={\psi}_i^{n},\\
\bar{\phi}_i^{\ell}=\bar{\psi}^{\ell}_{ i}, &\quad  {\phi}_i^{\ell}={\psi}^{\ell }_{i} \text{ for $\ell<n$.} 
\end{align}
Let $\nu^n_i=-\bar{\phi}_i^n{\phi}_i^n$.  Then $-\frac{1}{2} (\n_i-\n_j\,,\n_i-\n_j)$ transforms into
\begin{dmath}
\underbrace{-\sigma_i{(n-1)}\sigma_j{(n-1)} \left\{[\bar{\phi}^{n}_{ i}- \bar{\phi}^{n}_{j}][{\phi}^{n}_{i}-{\phi}^{n}_{ j}]+\nu_i^n \nu_j^n \right\}}_{\textrm{I}}+ \underbrace{1-\sigma_i{(n-1)}\sigma_j{(n-1)} +\sum_{\ell=1}^{n-1}\bar{\phi}^{\ell}_{ i}\cdot {\phi}^{\ell}_{j}}_{\textrm{II}}
\end{dmath}
and the \textit{apriori} integration form
$D\mu_{0}$ transforms to 
\[  \prod_{i} \prod_{\ell} \partial_{\bar{\phi^\ell_i}} \partial_{{\phi^{\ell}_i}}
    \circ {\sigma_i{(n-1)}}^{-3}[1-2\nu_i^n]^{-1/2}\;.
\]
The point here is that Term \textrm{I} above is, (ignoring the fact that the coefficient $\sigma_i{(n-1)}\sigma_j{(n-1)}$ is not a positive real number), algebraically of exactly the same form as
the coupling in an $\H^{0|2}$ model, while Term \textrm{II} is of the same form as that for an $\H^{0|2(n-1)}$ model. The caveats are that these models are coupled through the interaction $J_{ij} \sigma_i{(n-1)}\sigma_j{(n-1)}$ and that there a change in form of the \textit{a priori} measure due to the Jacobian $\sigma_i{(n-1)}^{-3}$.  By analogy with the language of probability theory, in the new variables the last component of spin is, `conditional' on the first $n-1$ components, generating a uniform spanning forest with `edge weights' $J_{ij} \sigma_i{(n-1)}\sigma_j{(n-1)}$.

By induction, we can choose a system of generators $\bar{\alpha}, \alpha$ for $\CG_{\Lambda}$ so that if we set $\mu_i^{k}=-\bar{\alpha}_i^k \alpha_i^k$
the action $S_{J, \epsilon}$ transforms into
\[
S_{J, \eps}\mapsto \tilde{S}_{J, \eps}=\sum_{(jj')} \sum_{\ell} \Gamma^{\ell}_{i} \Gamma^{\ell}_{j}\left\{[\bar{\alpha}^{\ell}_{ i}- {\bar{\alpha}}^{\ell}_{j}][{{\alpha}}^{\ell}_{i}-{{\alpha}}^{\ell}_{ j}]+{\mu}_i^{\ell}\bar{\mu}^{\ell}_j\right\}+\sum_{i} \eps_i  \Gamma^{N}_{i} 
\]
where
\[
\Gamma^{\ell}_{i}=\prod_{k<\ell} [1-\mu_i^{k}].
\]
Meanwhile, the integration form $D\mu_{0}$ transforms to 
\[ \prod_{i} \prod_{\ell} \partial_{{\bar{\alpha}_i^{\ell}}} \partial_{{{\alpha}_i^{\ell}}}
    \circ \frac{1}{\prod_{\ell=1}^{n} [1-2\mu_i^{\ell}]^{2(n-\ell)+1}}
\]

\begin{lemma}
\label{C:Poly}
For any $n\geq 1$ the partition function $Z^f_{J, \eps}$ is a multivariate polynomial of degree $n$ in the coupling constants $J_{ij}$.  Since $Z^f_{J, \eps}=Z_{J, \eps}$, the same holds for $Z_{J, \eps}$.
\end{lemma}
\begin{proof}
In the new variables $\bar{\alpha}, {\alpha}$, for any component $\ell$,
\[
\left\{[\bar{\alpha}^{\ell}_{ i}- {\bar{\alpha}}^{\ell}_{j}][{{\alpha}}^{\ell}_{i}-{{\alpha}}^{\ell}_{ j}]+{\mu}_i^{\ell}\bar{\mu}^{\ell}_j\right\}^2=0
\]
The lemma follows immediately from the pigeon hole principle.
\end{proof}

\subsection{Proof of Theorem \ref{T:Main2}, Positivity of the coefficients}
Now we turn to the positivity of $\partial_{J_{jj'}}^k Z_{J, \eps}$.
Though in the introduction part of the theorem is stated for general pinning fields $\eps$, for simplicity of exposition we will restrict attention to the case when this field is supported at one vertex, which we denote by $0$ and denote the field by $\eps_0$.  At the appropriate place, \Cref{Prop:Pos}, we indicate the changes that must be made to the statements for general pinnings.
Since  
\[
0 < Z_{J, \eps}= Z^f_{J, \eps} \text{ and }\frac{\partial_{J_{jj'}}^k Z^f_{J, \eps}}{Z^f_{J,\eps}}=\left \langle [-(\n_i-\n_j)^2]^{k} \right \rangle^f_{J,\eps},
\]
we shall rather estimate $\left \langle [-(\n_i-\n_j)^2]^{k} \right \rangle^f_{J,\eps_0}$ from below.
Before proceeding to the main line of proof, we are going to introduce a few calculational tools.

\subsubsection{Ward Identities for $\H^{0|2n}$ models}
\label{S:Ward}
Recall $\n_i=(\sigma_i, \psibar^{1}_i, \psi^1_i, \dotsc, \psibar^{n}_i, \psi^n_i)$.  Until now, we have introduced $\sigma_i$ as a shorthand for $\sqrt{1+2\sum_{\ell = 1}^n \psibar_i^{\ell}\psi_i^\ell}$.  In \Cref{S:R12n}, we discuss how the $\n_i$'s arise as vectors on (super)-hyperboloids in the superspace  $\BR^{1|2n}$ equipped with the Lorentz inner product
\begin{equation}
   v_i\cdot v_j
   \; := \;
- \sigma_i \sigma_j +(\psibar_i \psi_j - \psi_i \psibar_j)
   \;.
 \label{eq.scalarprod1}
\end{equation}
In particular, this linear form gives $n_i^2=-1$ and leads naturally to the objects we now introduce.
Let
\begin{eqnarray*}
  Q^{\ell}_+  & = &
    \sum_{i\in V}  \sigma_i \partial^{\ell}_i,
        \\[1mm]
  Q^{\ell}_-  & = &  \sum_{i\in V}  \sigma_i       \bar{\partial}_i,
     \label{eq:def_Q_bis}
\end{eqnarray*}
where 
$\partial^{\ell}_i = \partial_{\psi_i^{\ell}}$
and $\bar{\partial}^{\ell}_i = \partial_{\psibar_i^{\ell}}$.
We note the identities
\begin{eqnarray}
\label{E:SUSY41}
&Q^{\ell}_{+}{ {\psi}}_i^\ell=Q^{\ell}_{-}{ {\bar{\psi}}}_i^\ell= \sigma_i, &Q^{\ell}_{\pm}[ \n_i\cdot \n_j]=0.\\
\end{eqnarray}
For calculation purposes, we work with the unnormalized, unpinned state 
\[
\left \llbracket\cdot \right \rrbracket^f_{J, 0}:= \int  \CD(\psibar, \psi) \cdot e^{-S_{J, 0}}.
\]
Since $\left\{\sigma_i \sigma_j -\psibar_{ i}\cdot {\psi}_{j}-\psibar_{ j}\cdot {\psi}_{i}\right\}=-\n_i \cdot \n_j$
and the bare state is
\[
 \CD(\psibar, \psi)= \prod_{i, \ell} \partial_{\psibar_i^{\ell}} \partial_{\psi_i^{\ell}}\sigma_i^{-1},
\]   
the state $\left \llbracket\cdot \right \rrbracket^f_{J, 0}$  is invariant with respect to the infinitesimal symmetries $Q^{\ell}_{\pm}$ (see the discussion around \eqref{E:SUSY4}).  This fact implies a number of useful relations.  
The basic one is an integration-by-parts formula: 
\begin{lemma}
Expressing $\sigma_i=Q_-^\ell \bar{\psi}_i^{\ell} =Q_+ ^\ell \psi_i^{\ell}$,  for any $F\in \CG_{\Lambda}$
\beq
\label{E:IBP}
\left \llbracket  \sigma_iF(\bar{\psi}, \psi)  \right \rrbracket^f_{J,0}=\left \llbracket \psibar^{\ell}_i  \bar{\partial}^{\ell}_i F(\bar{\psi}, \psi) \right \rrbracket^f_{J,0}=\left \llbracket \psi^{\ell}_i  {\partial}^{\ell}_i F(\bar{\psi}, \psi) \right \rrbracket^f_{J,0}.
\eeq
\end{lemma}
Further discussion of \eqref{E:IBP} appears in \Cref{S:R12n}.
The next lemma records a special case of this identity which will be useful below.

\begin{lemma}[Expansion Identity]
Fix $L\leq n-1$ and suppose $F$ is an even element of the Grassmann algebra $\CG_{V}$ depending on the components $(\psibar_i^{\ell}, \psi_i^{\ell})_{\substack{\ell \leq L, \\ i \in \Lambda}}$.     Then 
\beq
\label{E:I1}
 \left \llbracket F \sigma_0^m e^{-\eps_0 \sigma_0} \right \rrbracket^f_{J,  0}=\eps_0 \left \llbracket F \pi_0^{L+1} \sigma_0^{m-1} e^{-\eps_0 \sigma_0}\right \rrbracket^{f}_{J,  0}  - (m-1) \left \llbracket F \pi_0^{L+1}\sigma_0^{m-2} e^{-\eps_0 \sigma_0} \right \rrbracket^{f}_{J,  0}.
 \eeq
 \end{lemma}
\begin{proof}
Since $F$ depends only on $(\psibar_i^{\ell}, \psi_i^{\ell})_{\substack{\ell \leq L, \\ i \in \Lambda}}$ \Cref{E:IBP} implies
\beq
\left \llbracket F  \sigma_0 e^{-\eps_0 \sigma_0}\right \rrbracket^{f}_{J, 0 }=\eps_0 \left \llbracket F  \pi_0^{L+1} e^{-\eps_0 \sigma_0}\right \rrbracket^{f}_{J,  0}
\eeq

\end{proof}

We are now ready to begin the main line of proof.  For the readers convenience, we restate, in the Grassmann language, what we are aiming to prove.
\begin{theorem}[Restatement  of \Cref{T:Main2} for $\eps_0\geq 0$ in Grassmann variables]
Fix $\eps_0>0$.  For all $J, k$ such that either $k\leq \lceil n/2 \rceil$ or $k=n$
\[
\left\langle [-(\n_i-\n_j)^2]^{k} \right\rangle^{f}_{J,   \eps_0}\geq 0.
\]
In addition, if $\eps_0\leq 1 $ then 
\[
\left\langle [-(\n_i-\n_j)^2]^{k} \right\rangle^{f}_{J,   \eps_0}\geq 0.
\]
for all $k$.
\end{theorem}

\textit{Proof.}
Let us consider first the case $k=n$, which we prove indirectly:
By \Cref{C:Poly}, as a function of $J_{ij}$ alone $Z_{J, \eps_0}$ is a polynomial of degree $n$, we can express this via the formula
\[
\frac{Z_{J, \eps_0}}{Z_{J, J_{ij}=0, \eps_0}}= \sum_{j=0}^n  \frac{J_{ij}^k}{k!} \langle [-(\n_i-\n_{j})^2]^k \rangle^{f}_{J, J_{ij}=0, \eps_0}
\]
(provided $Z_{J, J_{ij}=0, \eps_0}\neq 0$, a fact which is true whenever the weighted graph with $J_{ij}=0$ is connected and is then deduced through a limiting procedure otherwise).
Since the left hand side is non-negative for all $J_{ij}\geq 0$ we conclude that 
\[
\langle [-(\n_i-\n_{j})^2]^n \rangle^{f}_{J, J_{ij}=0, \eps_0}\geq 0.
\]
By Section \ref{S:UT} $ [-(\n_i-\n_{j})^2]^{n+1}=0$,
\[
\langle [-(\n_i-\n_{j})^2]^n \rangle^{f}_{J, J_{ij}=0, \eps_0}=\frac{Z_{J, \eps_0}}{Z_{J, J_{ij}=0, \eps_0}}\langle [-(\n_i-\n_{j})^2]^n \rangle^{f}_{J, \eps_0}
\]
so the claim holds for $k=n$ and all $J_{ij}\geq 0$.

For $k<n$ we need to get our hands dirty.  It will turn out that, for $k\leq n/2$, we can provide a proof which is relatively clean and easy to check. In particular it applies in the main practical application of interest since $\partial _{J_{ij}} \log Z_{J, \eps_0}=\left\langle [-(\n_i-\n_j)^2] \right\rangle^{f}_{J,   \eps_0}$.  For $n/2<k<n$ we currently have no better method than to express $\left\langle [-(\n_i-\n_j)^2]^{k} \right\rangle^{f}_{J,   \eps_0}$ as the expectation of a polynomial in $G_{00}, G^{(1)}_{ij}, G^{(2)}_{ij}, {G^{(3)}_{ij}}^2$ and to argue the the coefficients of the polynomial are positive.  This requires a fair amount of computational endurance.  

Recall the definition of $\tau^{\ell}_{ij}$ from \Cref{E:tau}.  Using  the symmetry of the state with respect to the Grassmann components and then \eqref{E:IBP} with $\sigma_i-\sigma_j=Q_-^1[\psibar^1_i-\psibar^1_j]$,
\begin{eqnarray}
\label{E:Q-}
\left \llbracket [-(\n_i-\n_j)^{2}]^{k} e^{-\eps_0 \sigma_0} \right \rrbracket_{J, 0}^f=& \left \llbracket 2n \tau_{ij}^1+(\sigma_i-\sigma_j)^2[-(\n_i-\n_j)^{2}]^{k-1}e^{-\eps_0 \sigma_0}  \right \rrbracket_{J, 0}^f.\nonumber\\
=& \left \llbracket (2n-1) \tau_{ij}^1- \eps_0 \eta_{ij}^1(\sigma_i-\sigma_j))[-(\n_i-\n_j)^{2}]^{k-1}e^{-\eps_0 \sigma_0}  \right \rrbracket_{J, 0}^f.
\end{eqnarray}

At this point let us indulge in a non-rigorous aside: If the pinning $\eps_0=0$ the second term in the second equality vanishes.  We can then iterate this identity using \eqref{E:IBP} with $\sigma_i-\sigma_j=Q_-^\ell[\psibar^\ell_i-\psibar^\ell_j]$ for $2\leq \ell \leq k$ successively in the analogous place where we had $\sigma_i-\sigma_j=Q_-^1[\psibar^1_i-\psibar^1_j]$.  We obtain, if $\eps_0=0$,
\[
\left \llbracket [-(\n_i- \n_j)^{2}]^k \right \rrbracket_{J, 0}^f=
[2(n-1)+1]\cdots [2(n-k)+1] \left \llbracket  \prod_{\ell=1}^k \tau_{ij}^\ell \right \rrbracket_{J, 0}^f.
\] 
valid for all $k\leq n$.  It is then tempting to use the duality \eqref{L:SUSY1} to conclude
\[
 \left \langle  \prod_{\ell=1}^k \tau_{ij}^\ell \right \rangle_{J, 0}^f``="  \left\langle [G^{(1)}_{ij}]^k     \right \rangle_{J, 0}\geq 0,
\]
but unfortunately the RHS is not well defined due to the lack of pinning for the $t$-field. Still, it provides motivation to push further.

Returning to \eqref{E:Q-}, we use $\sigma_i-\sigma_j=Q_-^2[\psibar^2_i-\psibar^2_j]$ to obtain two summands
\begin{align}
\left \llbracket [-(\n_i-\n_{j})^2]^k     e^{-\eps_0 \sigma_0}  \right \rrbracket_{J, 0}^f=&(2n-1)\left \llbracket [-(\n_i-\n_{j})^2]^{k-1}\tau^1_{ij}   e^{-\eps_0 \sigma_0}  \right \rrbracket_{J, 0}^f\nonumber\\
+&\eps_0^2\left \llbracket [-(\n_i-\n_{j})^2]^{k-1}[\psibar_i^1-\psibar_j^1]\psi_0^1 [\psibar_i^2-\psibar_j^2]\psi_0^2    e^{-\eps_0 \sigma_0}  \right \rrbracket_{J, 0}^f.
\end{align}

We want to continue iteratively, using that  $[\psibar_i^m-\psibar_j^m]^2= [\psi_i^m-\psi_j^m]^2 =0$ and  expanding $\sigma_i-\sigma_j$ using $Q_-^{\ell}$ until we generate $n$ distinct components between the $\tau_{ij}^\ell$'s and the $[\psibar_i^\ell-\psibar_j^\ell]\psi_0^\ell$'s.   To handle the combinatorics of this expansion let, for $p, t\in \N$
\begin{align}
&\CI_n(p, t)=\{(i_1, \dotsc, i_p): t\leq i_s\leq n, i_s>i_{s+1}, i_s-i_{s+1}=1 \text{ is odd}, n-i_1 \text{ is even}\},\\
&A_n(p, t)=\sum_{I\in \CI(p, t)}[2i_1-1]\cdots [2i_p-1],\\
&C_n(0)=1 \text{ and }C_n(p)=A_n(p,1) \text{ otherwise.  In particular, if $p>n C_n(p)=0$}\\
\end{align}
When $n$ and there is no danger of confusion we will suppress the subscript $n$. 

With these notations, and since the components of the Grassmann field are exchangeable, after iterating the previous computation using the available components $\ell=3,4, \cdots, n$ we have
\begin{align}
&\left \llbracket [-(\n_i- \n_j)^{2}]^{k} e^{-\eps_0 \sigma_0} \right \rrbracket_{J, 0}^f=\sum_{p=2k-n}^k A(p, n+p-2k)\eps_0^{2(k-p)} \left \llbracket  \prod_{\ell=1}^p \tau_{ij}^\ell \prod_{m=1}^{2(k-p)} \eta_{ij}^{m+p}     e^{-\eps_0 \sigma_0}  \right \rrbracket_{J, 0}^f \label{E:R}\\
&\quad \quad \quad \quad \quad \quad \quad \quad \quad \quad \quad \quad \quad \quad + \sum_{p=0}^{2k-n-1}C(p) \eps_0^{n-p}\left \llbracket  \prod_{\ell=1}^p \tau_{ij}^\ell \prod_{m=1}^{n-p} \eta_{ij}^{m+p}   [\sigma_i-\sigma_j]^{2k-n-p}  e^{-\eps_0 \sigma_0}  \right \rrbracket_{J, 0}^f\nonumber.
\end{align}
If $2k-n-1<0$, the second sum should be interpreted as $0$.
We pause now to observe an intermediate result.  In case $k\leq \lceil n/2\rceil$ the second term in \Cref{E:R} vanishes.  We claim the first term is manifestly positive.  To see this 
observe that if we absorb  $e^{-\eps_0 \sigma_0} $ into the action and pass back to horospherical coordinates, then we may express fermionic expectations for polynomials in the $\tau$'s, $\eta$'s and $\pi_0$'s in terms of polynomials in the ${G_{ij}^{(1)}}$, ${G_{ij}^{(2)}}$, ${G_{00}}$, and ${G^{(3)}_{ij}}$, see \Cref{L:GIBP}.
We therefore obtain the following statement:
\begin{lemma}
\label{Prop:Pos}
For any $k$, all summands in the first term on the RHS of \Cref{E:R} are positive.    In particular
if $k\leq \lceil n/2 \rceil$,
\[
\left \langle [-(\n_i- \n_j)^{2}]^k  \right \rangle^f_{J,\eps_0 }=\sum_{\substack{0\leq p, q: \\2p+ 2q=2k,\\p+2q\leq n}}A(p, n+p-2k) \eps_0^{2(k-p)} \left \langle {G^{(1)}_{ij}}^p{ G^{(3)}_{ij}}^{2q} \right \rangle_{J, \eps_0}\geq 0.
\]
In particular this shows that $Z_{J, \eps_0}^f=Z_{J, \eps_0}$ is increasing in the coupling constants.
\end{lemma}
\begin{remark}
To modify these formulas if the pinning field is supported at more than one point, replace $\eps_0 \eta_{ij}$ by
\[
\kappa_{ij}(\eps):=-[\psibar_i^{\ell}-\psibar_{j}^{\ell}][\sum_{j'} \eps_{j'} \psi_{j'}^\ell].
\]
This leads to a replacement of $\eps_0 G^{(3)}_{ij}$ by
\[
G^{(3)}_{ij}(\eps):=\sum_{j'} \eps_{j'}[G_{ij'}-G_{jj'}].
\]
\end{remark}

From now on, we may assume $\lceil n/2\rceil<k<n $. Combining this range with the positivity argument given for the $n$'th moment, we may also assume $n\geq 4$. From this point forward, the proof becomes more computational (i.e. less intrinsic).
For summands in the second term on the RHS of \Cref{E:R}, we continue by using $\sigma_i-\sigma_j=Q^{\ell}_+[\psi^{\ell}_i-\psi^{\ell}_j]$ for $\ell= p+1$.  Letting $\phi(j)= \sigma_0(j)- \eps_0\pi_0^j$,
\beq
 \left \llbracket  \prod_{\ell=1}^p \tau_{ij}^\ell \prod_{m=1}^{n-p} \eta_{ij}^{m+p}  [\sigma_i-\sigma_j]^{2k-p-n}   e^{-\eps_0 \sigma_0}  \right \rrbracket_{J, 0}^f= \left \llbracket \prod_{\ell=1}^{p+1} \tau_{ij}^{\ell} \prod_{m=2}^{n-p} \eta_{ij}^{m+p}   [\sigma_i-\sigma_j]^{2k-p-n-1} \phi(p+1) e^{-\eps_0 \sigma_0}  \right \rrbracket_{J, 0}^f
 \eeq
 where we used nilpotency of the $\bar{\psi}_0^\ell$'s and $\psi_0^\ell$'s to reduce $\sigma_0$ to $\sigma_0(p+1)$.
 Continuing in this way, for $\ell\in [p+2, 2k-n-p]$ we arrive at
\beq 
 \left \llbracket  \prod_{\ell=1}^p \tau_{ij}^\ell \prod_{m=1}^{n-p} \eta_{ij}^{m+p}  [\sigma_i-\sigma_j]^{2k-p-n}   e^{-\eps_0 \sigma_0}  \right \rrbracket_{J, 0}^f=\left \llbracket \prod_{\ell=1}^{2k-n} \tau_{ij}^{\ell} \prod_{m=2k-n+1}^{n} \eta_{ij}^{m}   \prod_{j=p+1}^{2k-n} \phi(j) e^{-\eps_0 \sigma_0}  \right \rrbracket_{J, 0}^f.
\eeq
Plugging this expression into \eqref{E:R} and reorganizing terms
\begin{align}
\label{E:Moment1}
&\left \llbracket [-(\n_i- \n_j)^{2}]^k e^{-\eps_0 \sigma_0} \right \rrbracket_{J, 0}^f=\underbrace{\sum_{\substack{0\leq p, q: \\2p+ 2q=2k,\\p+2q\leq n}}A(p, n-p-2q) \eps_0^{2q} \left \llbracket  \prod_{\ell=1}^p \tau_{ij}^\ell \prod_{m=1}^{2q} \eta_{ij}^{m+p}     e^{-\eps_0 \sigma_0}  \right \rrbracket_{J, 0}^f}_{\textrm{III}}\\
\nonumber
& \quad \quad \quad \quad + \mathbf 1\{2k\geq n\}\underbrace{\left \llbracket \prod_{\ell=1}^{2k-n} \tau_{ij}^{\ell} \times \prod_{m=2k-n+1}^{n} \eta_{ij}^{m}\times \left\{ \sum_{p=0}^{2k-n-1} C(p)\eps_0^{n-p}\prod_{s=p+1}^{2k-n}\phi(s)\right\}  e^{-\eps_0 \sigma_0}\right \rrbracket_{J, 0}^f}_{\textrm{IV}} .\nonumber
\end{align}

Using an integration by parts in horospherical coordinates, we may  estimate Term \textrm{III} from below by the  expectation of a polynomial (with positive coefficients) in $G_{00}, G^{(1)}_{ij}, G^{(2)}_{ij}$:   
\begin{lemma}
\label{C:T1}
We have, for $k\leq n-1$,
\beq
\frac{\textrm{III}}{Z_{J, \eps_0}}\geq \eps_0 A(k, n-k)\cdot  \sum_{\ell=0}^k D_n(\ell;k)  \left \langle G_{00}^{n-k} {G^{(1)}_{ij}}^{x}{G^{(2)}_{ij}}^{k-x} \right \rangle_{J, \eps_0},
\eeq
where
 \[
 D_n(\ell;k)\geq \begin{cases}
\frac{(n-k-1)!\cdot k! 2^{n-\ell-2} }{\ell!}{n-\ell-2 \choose k-\ell}. \text{ if $k\leq n-2$,} \\
  \frac{2^{n-\ell-2}\Gamma(n-\ell-1/2)(n-2)!}{\Gamma(n-\ell)\Gamma(1/2)\ell!} \text{ if $k=n-1$}
 \end{cases}
 \]
\end{lemma}

To  evaluate Term $\textrm{II}$ in a compact, implicit manner let 
\begin{align}
\label{E:D}&D_1=-\frac{b}{v }\partial _v+a \quad  D_2={\eps_0}\left[ a v-\frac{b}{v}-\eps_0b-b\partial_v\right]	\\
 &\CP_{L,m}(v; a, b)= D_1^{L} \cdot D_2^m \cdot 1. \end{align}
 Note that $\CP_{L, m}$ is \text{a rational function in $v$ and polynomial in $a, b$ of degree $L+m$}.
\begin{lemma}
\label{L:T2}
We have
\begin{align}
&\label{C:T2} \frac{\textrm{IV}}{Z_{J, \eps_0}}=\:\eps_0^{2(n-k)}\sum_{p=0}^{2k-n-1} C(p)\langle  [G^{(3)}_{ij} ]^{2(n-k)}\CP_{p,2k-n-p}(1; G^{(1)}_{ij}, G^{(2)}_{ij})\rangle_{\beta, \eps_0}\\
&\label{C:T3} \text{ and  if $0\leq  \eps_0\leq 1$ }  \left|\frac{\textrm{IV}}{Z_{J, \eps_0}}\right|\leq  \frac{\textrm{III}}{Z_{J, \eps_0}}
\end{align}
\end{lemma}

The proofs of \Cref{C:T1,L:T2} are presented in the next two subsections. However, let us remark that for $n/2<k\leq n-1$, Lemma \ref{C:T1} and Equation \eqref{C:T2} together show that under the weaker  assumption that $\eps_0$ is small as a function of $n$, Term \textrm{IV} can be bounded by Term \textrm{III}. To get a quantitative estimate of the dependence of $\eps_0$ on $n$, the main difficulty is to carefully understand $\CP_{p,2k-n-p}(1; a, b)$.  In particular, to emphasize the second statement $\left|\frac{\textrm{IV}}{Z_{J, \eps_0}}\right|\leq  \frac{\textrm{III}}{Z_{J, \eps_0}}$ for $\eps_0\leq 1 $ \textit{independent} of $n$.  This latter conclusion is demonstrated in \Cref{S:Comb}.  


\subsection{Proof of Lemma \ref{C:T1}}
\label{S:I}
\begin{proof}
Using \Cref{E:I1} we have the following identities:

\noindent
\textbf{Identity 1:}  If $F$ depends on at most the first $n-1$ components $(\bar{\psi}_i^{\ell}, {\psi}_i^{\ell})_{\substack{i \in V,\\ \ell \leq n-1}}$ then
\beq
\label{E:II1}
 \left \langle  F  \right \rangle^f_{J, \eps_0}=\eps_0 \left \langle  F   \pi_0^n \right \rangle^f_{J, \eps_0}+\sum_{j=1}^n  \left \langle  F \frac{\pi_0^j}{\sigma_0(j-1)} \right \rangle^f_{J, \eps_0}.
\eeq
This follows by applying \Cref{E:I1} to $ \left \langle  F \sigma_0 \right \rangle^f_{J, \eps_0}$ with $\sigma_0=Q^n_-\psibar ^n_{0}$.

\noindent
\textbf{Identity $2$:} If $F$ depends on, at most, the first $n-2$ components $(\bar{\psi}_i^{\ell}, {\psi}_i^{\ell})_{i \in V, \ell \leq n-2}$
\beq
\label{E:II2}
 \left \langle  F  \right \rangle^f_{J, \eps_0}=\eps_0^2 \left \langle  F   \pi_0^{n-1} \pi_0^n \right \rangle^f_{J, \eps_0}+\left \langle  F   \pi_0^n \right \rangle^f_{J, \eps_0} + 2 \sum_{j=1}^{n-1}  \left \langle  F {\pi_0^j} \right \rangle^f_{J, \eps_0}.
\eeq
This follows similarly by applying \Cref{E:I1} to $ \left \langle  F \sigma_0^2 \right \rangle^f_{J, \eps_0}$ with 
\[
\sigma_0^2+\pi_0^n=Q^n_-\cdot [\bar{\psi}^n_0 \sigma_0]=Q^n_-\cdot\left[ \bar{\psi}^n_0 \cdot [Q^{n-1}_- \cdot \bar{\psi}^{n-1}_0]\right].
\]

To prove the lemma, we distinguish two cases:  Either $\ceil{n/2} \leq k\leq n-2$ or $k=n-1$.  In both cases, we shall start by applying these identities with $F=\prod_{\ell=1}^k\tau_{ij}^{\ell}$.  Note that we may ignore the remaining contributions  to Term \textrm{III} since by Proposition \ref{Prop:Pos} they are positive.   
A convenient identity to use when working with these identities is
\[
\frac{1}{\sigma_0(j)}=\E[e^{X^2\sum_{\ell\leq j} \pi_0^\ell}]
\]
where $ X$ denotes a  $N(0,1))$ random variable and $\E$ refers to integration with respect to $X$. Then after horospherical integration, each summand on the RHS of \Cref{E:II1,E:II2} is non-negative (if this is too brief, the reader may look to the case $k=n-1$ below for more detail).

We take up the case $k\leq n-2$ first.
Since $\left\langle \prod_{\ell=1}^k\tau_{ij}^{\ell}   \pi_0^{n-1} \pi_0^n \right \rangle^f_{J, \eps_0}\geq 0$, \Cref{E:II2} implies
\[ 
\left \langle  \prod_{\ell=1}^k\tau_{ij}^{\ell}  \right \rangle^f_{J, \eps_0}\geq 2 \sum_{j=1}^{n-1}  \left \langle  \prod_{\ell=1}^k\tau_{ij}^{\ell} {\pi_0^j} \right \rangle^f_{J, \eps_0}
\]
We repeatedly apply \Cref{E:II2} to each term on the RHS until we reach an expression which depends on $n-k-1$ of the last $n-k$ components.
We obtain (using the exchangeability of the components)
\beq
\label{E:ee1}
\left \langle  \prod_{\ell=1}^k\tau_{ij}^{\ell}  \right \rangle^f_{J, \eps_0}\geq \sum_{J\leq k} b_J(k) 2^{J+n-k-1} \langle  \prod_{\ell=1}^k\tau_{ij}^{\ell} \prod_{m=k-J+1}^{n-1} \pi_0^m\rangle^f_{J, \eps_0}
\eeq
where 
\[
b_J(k)={J+n-k-2 \choose J} \frac{(n-k-1)!\cdot k!}{2\cdot (k-J)!}
\]
represents the number of ways of producing  a product $ \prod_{i=1}^{J+n -k-1} \pi_0^{m_i} $ such that exactly $J$ of the $m_i$'s are less or equal to  $k$ and the index $m_{J+n -k-1}>k$.

Once $n-k-1$ of the last $n-k$ components appear, we switch protocols and use \Cref{E:II1} to each summand on the RHS. 
We obtain, since the remarks on positivity remain in force for 
\[
F:=\prod_{\ell=1}^k\tau_{ij}^{\ell} \prod_{m=k-J+1}^{n-1} \pi_0^m
\] 
as well,
\[
\left \langle  \prod_{\ell=1}^k\tau_{ij}^{\ell} \prod_{m=k-J}^{n-1} \pi_0^m\right \rangle^f_{J, \eps_0}\geq \sum_{k-J \geq l} \frac{(k-J)!}{l!} \left \langle  \prod_{\ell=1}^k\tau_{ij}^{\ell} \prod_{m=l+1}^{n} \pi_0^m \right\rangle^f_{J, \eps_0}.
\]
Plugging this into the RHS of \Cref{E:ee1}, we have
\[
\left \langle  \prod_{\ell=1}^k\tau_{ij}^{\ell}  \right \rangle^f_{J, \eps_0}\geq \sum_{l\leq k} D_n(l;k)\left \langle  \prod_{\ell=1}^k\tau_{ij}^{\ell} \prod_{m=l+1}^{n} \pi_0^m\right \rangle^f_{J, \eps_0}
\]
where
\[
D_n(l;k)=  \frac{(n-k-1)!\cdot k!  }{l!}\sum_{J\leq k-l} {J+n-k-2 \choose J} 2^{J+n-k-2}. 
\]
Using \Cref{L:GIBP}, we have 
\[
\left \langle  \prod_{\ell=1}^k\tau_{ij}^{\ell} \prod_{m=k-l}^{n} \pi_0^m\right \rangle^f_{J, \eps_0}=\left \langle [G_{00}]^{n-k} {G^{(1)}_{ij}}^l {G^{(2)}}_{ij}^{k-l} \right \rangle_{J, \eps_0}
\]
Finally we have an appropriate lower bound on $D_n(l;k)$ simply by taking the term with largest index in the sum over $J$,
\[
D_n(l;k)\geq \frac{(n-k-1)!\cdot k! 2^{n-l-2} }{l!}{n-l-2 \choose k-l}.
\]

Turning to the case $k=n-1$, we insert the identity
\[
\frac{1}{\sigma_0(j)}=\E[e^{X^2\sum_{\ell\leq j} \pi_0^\ell}] \quad  (X\stackrel{d}{=} N(0,1)),
\]
to obtain
\[
\left \langle  \prod_{\ell=1}^{n-1}\tau_{ij}^{\ell}  \right \rangle^f_{J, \eps_0}=\left \langle  \pi_0^n \cdot \prod_{\ell=1}^{n-1}\tau_{ij}^{\ell}  \right \rangle^f_{J, \eps_0}+\sum_{j=1}^{n-1}\E\left[\left \langle \prod_{\ell=1}^{n-1}\tau_{ij}^{\ell} \cdot  \pi_0^j \cdot e^{X^2[ \sum_{\ell\leq j-1} \pi_0^\ell]} \right \rangle^f_{J, \eps_0}\right].
\]
By nilpotency of the $\pi_0^\ell$'s,
\[
e^{X^2[ \sum_{\ell\leq j} \pi_0^\ell]}=\sum_{J\leq j} \frac{X^{2J}}{J!} \left( \sum_{\ell\leq j} \pi_0^\ell \right)^J.
\]
Collecting terms according to the number of $\pi_0^{\ell}$'s appearing,
\[
\left \langle  \prod_{\ell=1}^{n-1}\tau_{ij}^{\ell}  \right \rangle^f_{J, \eps_0}=\left \langle  \pi_0^n \prod_{\ell=1}^{n-1}\tau_{ij}^{\ell}  \right \rangle^f_{J, \eps_0}+\sum_{m=0}^{n-1} a_m \left \langle   \prod_{\ell=1}^{n-1}\tau_{ij}^{\ell}  \prod_{\ell'=1}^m  \pi_0^{\ell'}\right \rangle^f_{J, \eps_0}
\]
where an empty product is treated as $1$ and
\[
a_{m}= \frac{[2(m-1)]!!}{(m-1)!}\sum_{m\leq j\leq n-1} \frac{(j-1)!}{(j-m)!}.
\]
This identity can be recursed for each summand such that $\pi_0^n$ does not appear.
We obtain
\[
\left \langle  \prod_{\ell=1}^{n-1}\tau_{ij}^{\ell}  \right \rangle^f_{J, \eps_0}=\sum_{m=0}^{n-1} D_{n}(n-m-1; n-1) \left \langle  \pi_0^n \prod_{\ell=1}^{n-1}\tau_{ij}^{\ell}  \prod_{\ell'=1}^m  \pi_0^{\ell'}\right \rangle^f_{J, \eps_0}.
\]
Here, with $j_0:=0$ and, for a given $i$-tuple of positive integers $j_1\dotsc j_i$ denoting $J_{i}=\sum_{t=0}^{i}j_t$,
\[
D_{n}(n-m-1; n-1)=\sum_{s\geq 0} \sum_{\substack{j_1 \dotsc  j_s\geq 1,\\ \sum_{i=1}^s j_i=m}} \prod_{i=0}^{s-1} a_{j_{i+1}, n-J_{i}-1}.
\]
We wish to estimate $D_{n}(n-m-1; n-1)$ from below.  The simplest option is to take the term $a_{m, n-1}$ to obtain
\[
D_{n}(n-m-1; n-1)\geq  \frac{2^{m-1}\Gamma(m+1/2)(n-2)!}{\Gamma(m+1)\Gamma(1/2)(n-m)!}
\]
\end{proof}

\subsection{Proof of Lemma \ref{L:T2}}
\label{S:II}
The starting point is the formula for Term \textrm{IV} of \Cref{E:Moment1}, which we recall here for the readers convenience:
\beq
\label{E:I6}
\frac{\textrm{IV}}{Z_{J\eps_0}}=\left \langle \prod_{\ell=1}^{2k-n} \tau_{ij}^{\ell} \times \prod_{m=2k-n+1}^{n} \eta_{ij}^{m}\times \left\{ \sum_{p=0}^{2k-n-1} C(p)\eps_0^{n-p}\prod_{m'=p+1}^{2k-n}\phi(m')\right\}  \right \rangle_{J, \eps_0}^f.
\eeq
To derive identity \eqref{C:T2}, on the RHS of \eqref{E:I6} we pass back to the horospherical representation and perform a fermionic Gaussian integration.  
In each summand, we first integrate out the components $\ell\in \{2k-n+1, \cdots, n\}$.  This yields a factor $[{G^{(3)}_{ij}}]^{2(n-k)}$.  Next we integrate the remaining components starting with $\ell=2k-n$ first and then proceeding backwards to $\ell=1$.  As a result of the component by component integration we have, after integrating out the components $2k-n, \cdots, s+1$, a recursively defined expression depending on the components $\ell=1, \dotsc, s$.  The expression factors as a product of two terms: an explicit  factor $T_s:=\prod_{m=p+1}^{s-1}\phi(m)$ independent of the $s^{th}$ component and another factor $Q_s$ which collects the dependence on the $s$'th component explicitly and through $\sigma(s)$.  To describe $Q_s$, recall $D_1, D_2$ from \Cref{E:D}.
With them evaluated at $a=G^{(1)}_{ij}, b=G^{(2)}_{ij}$ we compute that  
\[
Q_s=\begin{cases}
[{G^{(3)}_{ij}}]^{2(n-k)} \prod_{m=p+1}^s \phi(m) \times [D_2^{2k-n-s}\cdot 1|_{v=\sigma(s)}] \text{ if $s\geq p+1$,}\\
[{G^{(3)}_{ij}}]^{2(n-k)} \tilde{D}_1^{p-s} \tilde{D}_2^{2k-n-p}\cdot 1|_{v=\sigma(s)}]  \text{ if $s\leq p$.}
\end{cases}
\]
This is to be integrated with respect to the $t$-field and the remaining Gaussian components - there is $1$ real component and $s+1$ fermionic components.   The additional Gaussian components, 1 bosonic and 2 fermionic,  are due to the inversion of SUSY localization.  Note that the additional Gaussian components \textit{do not} appear in the delocalized versions of $T_s, Q_s$ since we may delocalize without
The first part of Lemma \ref{L:T2} thus follows. 

\subsubsection{Computing the $\CP$'s}
\label{S:Comb}
For the second part of Lemma \ref{L:T2},
let us consider 
\[
\CQ(v, a, b) :=\sum_{p=0}^{2k-n-1} C(p) \CP_{p,2k-n-p}(v; a, b).
\]
as remarked before,  this is a polynomial in $a, b$. 
Let $r=a/b$.  Then setting $u=\sqrt{r} v$, we have
\begin{align*}
&D_1=\frac{b r}{u} [u-\partial_u],\\
&D_2=\frac{\eps_0b\sqrt{r}}{u} [u - \partial_u-\frac{\eps_0}{\sqrt{r}} ]\cdot u.
\end{align*}
This leads to
\[
D_2^m \cdot 1=[\eps_0b\sqrt{r}]^m \left\{ H_{m+1}(u -\frac{\eps_0}{\sqrt{r}})+\frac{\eps_0}{\sqrt{r}} H_{m}(u -\frac{\eps_0}{\sqrt{r}})\right\},
\]
where $H_r$ are the Hermite polynomials.
Therefore 
\[
\CQ(v; a, b) =\sum_{p=0}^{2k-n-1} C(p)[\eps_0b\sqrt{r}]^{2k-n-p} [br]^p D_1^p \cdot \frac 1u \left\{H_{2k-n-p+1}(u -\frac{\eps_0}{\sqrt{r}})+\frac{\eps_0}{\sqrt{r}} H_{2k-n-p}(u -\frac{\eps_0}{\sqrt{r}})\right\}
\]
Recall that 
\[
H_m(z+y)=\sum_{j=0}^m {m \choose j}z^j H_{m-j}(y)
\]
so
\begin{align*}
&\CQ(1; a, b) =\\
&b^{2k-n}\sum_{p=0}^{2k-n-1} \sum_{j=0}^{2k-n-p+1} C(p)\eps_0^{2k-n-p}  \sqrt{r}^{2k-n-p+j-1}{2k-n-p+1 \choose j}2^p \left(\frac{j-1}{2}\right)_p H_{2k-n-p+1-j}(-\frac{\eps_0}{\sqrt{r}}) \\
&-b^{2k-n} \sum_{p=0}^{2k-n-1} \sum_{j=0}^{2k-n-p} C(p) \eps_0^{2k-n-p+1}   \sqrt{r}^{2k-n-p+j-2}{2k-n-p \choose j}2^p \left(\frac{j-1}{2}\right)_p H_{2k-n-p-j}(-\frac{\eps_0}{\sqrt{r}})
\end{align*}
where $(x)_j=x(x-1)\cdots(x-j+1)$ is the descending factorial and if $j=0$, we set $(x)_j=1$. 
Now we collect terms according to powers in $r$.  Expanding the Hermite polynomials using the variable $m$, and denoting $t_p=2k-n-p$,
\begin{align*}
\CQ(1; a, b) =&b^{2k-n}\sum_{p=0}^{2k-n-1} \sum_{j=0}^{t_p+1} \sum_{m=0}^{t_p+1-j} (-1)^m c_{m;t_p+1-j} C_n(p)\eps_0^{t_p+m}  \sqrt{r}^{t_p+j-1-m}{t_p+1 \choose j}2^p \left(\frac{j-1}{2}\right)_p  \\
&-b^{2k-n} \sum_{p=0}^{2k-n-1} \sum_{j=0}^{t_p} \sum_{m=0}^{t_p-j}  (-1)^m  c_{m;t_p-j}C_n(p) \eps_0^{t_p+m+1}   \sqrt{r}^{t_p+j-2-m}{t_p \choose j}2^p \left(\frac{j-1}{2}\right)_p
\end{align*}
where
\[
c_{w;r}=\begin{cases}
0 \quad \text{ if } r-w \text{ is odd or $w<0$,}\\
\frac{r!}{w!2^{(r-w)/2}\frac{r-w}{2}!}\quad  \text{ else.}
\end{cases}
\]

Recalling that $\CQ(1; a, b)$ must be a polynomial of homogenous degree $2k-n$, the only terms which ultimately contribute satisfy the conditions
\begin{align*}
0 \leq t_p+j-1-m\leq 2k-n, \quad t_p+j-1-m\in 2\Z \text{ from the first sum,}\\
0 \leq t_p+j-2-m\leq 2k-n, \quad t_p+j-2-m\in 2\Z \text{ from the second sum.}
\end{align*}
So
\[
\CQ(1; a, b)=\sum_{\ell=0}^{2k-n-1} d_n(\ell;k) a^{\ell} b^{2k-n-\ell}
\]
where
\begin{align}
\label{E:Bound1}&d_n(\ell;k)=C_n(2k-n-\ell)\eps_0^{\ell}  2^{2k-n-\ell} \left(\frac{\ell}{2}\right)_{2k-n-\ell}+ \\
&\sum_{p=0}^{2k-n-1} \sum_{j=0}^{t_p}(-1)^{t_p+j-1} C_n(p)\eps_0^{2t_p+j-1-2\ell}  2^p \left(\frac{j-1}{2}\right)_p
\nonumber \left\{ c_{t_p+j-2-2\ell;t_p-j} {t_p \choose j} + c_{t_p+j-1-2\ell;t_p+1-j} {t_p+1 \choose j}\right\}.
\end{align}

The proof of Lemma \ref{L:T2} is then concluded by verifying the following estimates.
\begin{lemma}
\label{L:CB}
For all $n\geq 4$ and all $\ceil{n/2}\leq k\leq n-1, 0\leq \ell \leq 2k-n-1$ and all $|\eps_0|\leq 1$
 \beq
 \label{E:A1}
A(k, n-k) D_{n-k+\ell} (k)\geq  |d_n(\ell;k)|
\eeq
\end{lemma}

To prove this lemma, we first estimate the growth of the $C_n(p)$'s.
\begin{lemma}
\label{L:Cp}
Let $n\geq 4$ be fixed.  We have
\begin{align}
\label{E:Cp}
&C_n(p)\geq [2(n-p)-1] C_{n}(p-1) \quad  \text{ for all $0 \leq p \leq n-3$,}\\
\nonumber &C_n(n)=C_n(n-1)\geq 2C_n(n-2),\\
\nonumber &C_n(n-2)\geq  \frac{15}{7} C_n(n-3).
\end{align}
\end{lemma}
\begin{proof}
The conditions on the $p$-tuples in $\CI_n(p, 1)$
imply that
\[
C_n(p)=(2n-1)C_{n-1}(p-1)+C_{n-2}(p)
\]
By induction, we see immediately that $C_n(p)$ increases in $p$.
To prove the lemma we induct on hypothesis that the stated bound holds for $n'< n$ and all $p$ and also for $n$ and $p'<p$.  Recall that $C_n(0)=1$, $C_n(1)\geq 2n-1$, so the bound clearly holds with $p=1$.

Before proceeding to verify the induction step, we need to make some observations:  First, by definition of the sets $\CI_n(p)$,
\begin{align}
&\nonumber C_n(n)=C_n(n-1)=A_n(1)=A_n(3)\\
&\label{E:n-n-2} C_n(n-2)=C_{n}(n) \sum_{j=2}^{n}[2j-1]^{-1}[2j-3]^{-1}\leq \frac 12 C_n(n)
\end{align}
since we have the classical identity
\[
\sum_{j=2}^{n}[2j-1]^{-1}[2j-3]^{-1}=\frac 12 -\frac 1{2(2n-1)}.
\]
Second, by a simple induction $C_n(p)\geq 2 C_n(p-1)$ for all $n\geq 2$ and $p\leq n-3$.
To verify this, we check by hand that 
\[
C_2(2)=C_2(1)=3, C_3(3)=C_3(2)= 15, C_3(1)=6. 
\] Also, by definition of $\CI_n(p)$,
\beq
\label{Cn-2}
C_n(n-2)- C_{n}(n-3)=\frac{4 C_n(n)}{3\cdot 5}
\eeq
so that we have
\[
C_n(n-2)= \frac{1-\frac 1{2n-1}}{1-\frac 1{2n-1}-\frac{8}{15}} C_n(n-3).
\]
This implies $C_n(n-2)\geq  \frac{15}{7} C_n(n-3)$ for all $n\geq 4$ (and this bound also holds for $n=3$).

We are now ready to verify the  induction step in the proof of the main lemma. We induct on hypothesis that the stated bound holds for $n'< n$ and all $p$ and also for $n$ and $p'<p$.  Then we have
\begin{align}
\nonumber C_{n}(p)=& (2n-1)C_{n-1}(p-1)+C_{n-2}(p)\\
\nonumber {\geq}& (2n-1)[2(n-p)-1]C_{n-1}(p-2)+[2(n-p-2)-1]C_{n-2}(p-1)\\
\label{E:Cn}=& [2(n-p)-1] C_{n}(p-1)+[2(n-p)-1][C_{n-2}(p)-C_{n-2}(p-1)]-2C_{n-2}(p)
\end{align}
where the inequality follows from the induction hypothesis.  Now since $p\leq n-3$, $n-p\geq 3$ so that $2(n-p)-1\geq 5$.  Also, by the \textit{a priori} estimate $C_{n-2}(p)\geq 2C_{n-2}(p-1)$.  It follows that the last two terms on the RHS of \eqref{E:Cn} sum to something non-negative.  Thus the induction step is proved.
\end{proof}

\begin{proof}[Proof of Lemma \ref{L:CB}]
Let 
$$b_n(\ell,k)=C_n(2k-n-\ell)\eps_0^{\ell}  2^{2k-n-\ell} \left(\frac{\ell}{2}\right)_{2k-n-\ell}.$$
Recall the formula for the Hermite coefficients $c_{r;s}$ which provide also constraints on the summands of $d_{\ell}(k)$ to be non-zero.
By the triangle inequality, the assumption $|\eps_0|\leq 1$ and Lemma \ref{L:Cp}
\[
|d_n(\ell;k)-b_n(\ell,k)|\leq  \sum_{p=0}^{2k-n-\ell-1}\sum_{j=[2\ell+1-t_p]\vee 0}^{\ell+1} 2^{2p-n+4}  \frac{ C_n(n-3)\Gamma(\frac 52)|((j-1)/2)_p|(t_p+1)!}{([t_p+j-2-2\ell]\vee 0)! (\ell+1-j)!j!\Gamma(n-p-1/2)}.
\]
With foresight on what our final comparison will be, we also bound
\[
 \frac{1}{(\ell+1-j)!j!}\leq 2^{\ell+1}\frac 1{(\ell+1)!}.
\]
to obtain
\[
|d_n(\ell;k)-b_n(\ell,k)|\leq   \sum_{p=0}^{2k-n-\ell-1}\sum_{j=[2\ell+1-t_p]\vee 0}^{\ell+1} 2^{2p+\ell-n+5}  \frac{C_n(n-3)\Gamma(\frac 52)|((j-1)/2)_p|(t_p+1)!}{(t_p+j-1-2\ell)!(\ell+1)!\Gamma(n-p-1/2)}.
\]

To estimate the RHS of this inequality, observe that for $j$ fixed 
\[
\frac{1}{(t_p+j-2-2\ell)!}, \quad \frac{(t_p+1)!}{\Gamma(n-p-1/2)}
\]
are both increasing in $p$. Denoting $M_\ell=|(\ell/2)_{2k-n-\ell}|\vee  |((\ell+1)/2)_{2k-n-\ell}|$ we also have
\[
|((j-1)/2)_p| \leq 4M_{\ell}
\]
for $j \leq \ell+1$ and $p\leq 2k-n+j-1-2\ell$. 
Exchanging the sums over $j$ and $p$, these bounds imply
\[
\nonumber |d_n(\ell;k)-b_n(\ell,k)|\leq  \label{E:S1}\sum_{ j=0}^{\ell+1}  2^{4k-3n+2j-3\ell+7}  \frac{C_n(n-3){\Gamma(\frac 52)}M_{\ell} (2\ell+1-j)!}{(\ell+1)!\Gamma(2(n-k)+2\ell-j+1/2)}.
\]
Since 
\[
\frac{(2\ell+1-j)!}{\Gamma(2(n-k)+2\ell-j+1/2)}
\]
is increasing in $j$, the RHS of \eqref{E:S1} can be bounded by
\[
 |d_n(\ell;k)-b_n(\ell,k)|\leq  2^{4k-3n-\ell+7}  \frac{C_n(n-3){\Gamma(\frac 52)}M_{\ell}}{(\ell+1)\Gamma(2(n-k)+\ell-1/2)}
\]
Finally, by Lemma \ref{L:Cp},
\[
|b_n(\ell,k)|\leq  2^{-2(n-k)+3-\ell} \frac{C_n(n-3)M_{\ell} \Gamma(\frac 52)}{\Gamma(2(n-k)+\ell-1/2)}.
\]

Let us now check the statement of the Lemma. For $k\leq n-2$ we can write
\[
D_n(\ell;k) A(k, n-k)\geq \frac{(n-k-1)!\cdot k! 2^{k+n-\ell-2} }{\ell!}{n-\ell-2 \choose k-\ell}\frac{\Gamma(n+1/2)}{\Gamma(n-k-1/2)}
\]
where we used Lemma \ref{C:T1} and
\[
A(k, n-k)=\frac{2^k \Gamma(n+1/2)}{\Gamma(n-k-1/2)}.
\]
We thus have
\[
\frac{|b_n(\ell;k)|}{D_n(\ell;k) A(k, n-k)}\leq 2^{k-3n+5} \frac{\Gamma(\frac 52)C_n(n-3){\Gamma(n-k-1/2)}}{\Gamma(n-k)\cdot k!\cdot  {n-\ell-2 \choose k-\ell} \Gamma(n+1/2)}   \frac{\ell!M_{\ell}}{\Gamma(2(n-k)+\ell-1/2)}
\]
and
\[
\frac{|d_n(\ell;k)-b_n(\ell,k)|}{D_n(\ell;k) A(k, n-k)}\leq 2^{3(k-n)-n+9 }\cdot \frac{\Gamma(\frac 52)C_n(n-3)\Gamma(n-k-1/2)}{\Gamma(n-k) k!{n-\ell-2 \choose k-\ell}\Gamma(n+1/2)} \frac{\ell!M_{\ell}}{(\ell+1)\Gamma(2(n-k)+\ell-1/2)}.
\]
We combine the two previous estimates with the fact that
\[
2^{3-n}C_n(n-3)\leq \frac{\Gamma(n+1/2)}{\Gamma(1/2)}
\]to obtain
\[
\frac{|d_n(\ell;k)|}{D_n(\ell;k) A(k, n-k)}\leq \frac{\Gamma(\frac 52)}{\Gamma(n-k)\cdot k!\cdot  {n-\ell-2 \choose k-\ell} }   \frac{ \ell!M_{\ell}}{\Gamma(2(n-k)+\ell-1/2)}\left\{ 2^{k-2n+2}+ \frac{2^{3(k-n)+6 }}{\ell+1}\right \} .
\]
Continuing with the RHS, we have
\begin{multline}
 \frac{\Gamma(\frac 52)}{\Gamma(n-k)\cdot k!\cdot  {n-\ell-2 \choose k-\ell} }   \frac{ \ell!M_{\ell}}{\Gamma(2(n-k)+\ell-1/2)}\left\{ 2^{k-2n+2}+ \frac{2^{3(k-n)+6 }}{\ell+1}\right \} \\
 \leq  \frac{\Gamma(\frac 52)}{\Gamma(n-k)\cdot k! } \frac{ (2k-n)!}{\Gamma(2(n-k)-1/2)}\left\{ 2^{k-2n+2}+ \frac{2^{3(k-n)+6 }}{\ell+1}\right \} 
 \leq  \frac{ (2k-n)!}{{ 5 \cdot k! }}\left\{ 2^{-n}+ 1\right \}\leq 1
\end{multline}
where in the last inequality we used the fact that $n\geq 4$ and $\ceil{n/2}\leq k\leq n-2$.
This verifies the Lemma in case $k\leq n-2$.  

If $k=n-1$, the estimate proceeds in the same way except that we replace the lower bound on $D_n(\ell;k)$ by
\[
D_n(\ell, n-1)\geq \frac{2^{n-\ell-2}\Gamma(n-\ell-1/2)(n-2)!}{\Gamma(n-\ell)\Gamma(1/2)\ell!}.
\]
The estimate gets slightly tighter.
We find
\[
 |d_n(\ell, n-1)-b_n(\ell,n-1)|\leq  2^{n-\ell+1}  \frac{C_n(n){\Gamma(\frac 52)}M_{\ell}}{(\ell+1)\Gamma(2+\ell-1/2)}
\]
and 
\[
|b_n(\ell,n-1)|\leq  2^{1-\ell} \frac{C_n(n)M_{\ell} \Gamma(\frac 52)}{\Gamma(2+\ell-1/2)}.
\]
We end up with
\begin{align*}
\frac{|d_n(\ell;n-1)|}{D_n({\ell},n-1) A(n, 1)}\leq & \frac{3}{2}  \frac{\sqrt{n-\ell} M_{\ell}}{(n-2)!\sqrt{\ell+1/2}}\left\{2^{2-n} +\frac{2^3}{\ell+1} \right\}\\
\leq & \frac{3}{8}  \frac{\sqrt{2n}}{(n-2)(n-3)}\left\{2^{2-n} +2^3\right\}
\end{align*}
For $n\geq 6$ the last expression is bounded by $1$. For $n\in \{4, 5\}$ and $k=n-1$, one must unfortunately return to \eqref{E:Bound1} and compute explicitly.  We find
\[
\begin{array}{lll}
|d_4(0;3)|\leq 5/2+2^6;
&|d_4(1;3)|= 0;&\\
A(4, 1)= 7\cdot 5\cdot 3;&&\\
|d_5(0;4)|\leq 3^3+5^2 \cdot 4;
&|d_5(1;4)|\leq 2^5 ;
&|d_5(2;4)|= 0;\\
A(5, 1)=3^3\cdot 7\cdot 5.&&
\end{array}
\]
Since $d_{n}(\ell; n-1)>1$ in the nonzero cases, these estimates complete the Lemma in the final two cases.
\end{proof}

\section{Bounds on Fourier Transforms:  Proof of Theorem \ref{T:Main}}
\label{S:Fourier}
In the setting of Theorem \ref{T:Main}, $J_{jj'}=\beta I_{jj'}$, $\rho_j$ is any function which is $j=0$ at $0$ and $\min_{(jj')} \cos(\rho_j-\rho_{j'})>0$.  We will shortly specify a particular choice by optimizing an upper bound on the Fourier transform of $t_m-t_{\ell}$.  By \Cref{L:FT} and \Cref{T:Main2}, we have
\[
\left|\left \langle e^{ik(t_m-t_\ell)} \right \rangle_{\Lambda, a, J, \varepsilon_0} \right|\leq \prod_{(jj')} \cos(\rho_j-\rho_j')^{{-a}} e^{ \sum_{(jj')\in \Lambda } J_{jj'}[1- \cos(\rho_j-\rho_j')]}  e^{-k[\rho_m-\rho_\ell]}.
\]
Using the bound $ 1-x^2/2\leq \cos(x) $,
\[
\left|\left \langle e^{ik(t_m-t_\ell)} \right \rangle_{\Lambda, a, J, \varepsilon_0} \right|\leq \prod_{(jj')} \cos(\rho_j-\rho_j')^{{-a}} e^{ \sum_{(jj')\in \Lambda } \beta (\rho_j-\rho_j')^2/2}  e^{-k[\rho_m-\rho_\ell]}.
\]
Since $\cos(x)\geq 1-x^2/2$ and using the bound $\frac{1}{1-u}\leq e^{bu}$ if $b>1$ and  $0\leq u< \frac{b-1}{b}$, we have
\[
\left|\left \langle e^{ik(t_m-t_\ell)} \right \rangle_{\Lambda, a, J, \varepsilon_0} \right|\leq e^{ \sum_{(jj')\in \Lambda } [\beta+ba] (\rho_j-\rho_j')^2/2}  e^{-k[\rho_m-\rho_\ell]}.
\]
if we restriction attention to  $\rho$'s so that $\|\nabla\rho \|_{\infty}<\sqrt{\frac{b-1}{b}}$.  

We would like to optimize over $\rho$, except that for $k$ large one of these optimizers - approximately a solution to
\[
-\Delta_{\Lambda}\cdot  \rho=\frac{k}{\beta+a}[\delta_m-\delta_\ell] 
\]
- may not satisfy $\cos(\nabla \rho)>0$ or $\|\nabla\rho \|_{\infty}<\sqrt{\frac{b-1}{b}}$.

To get around this issue, we cut the expected optimizer off.  Let $c=\frac 12\sqrt{\frac{b-1}{b}}$and let 
\[
B_x=\{j: \frac{|k|}{(\beta+ba)}<c \|x-j\|_{2}\}.
\]  
Given $k$, define
\begin{align}
&\phi_j(k;x)=
-\frac{k}{\beta+ba}\log(1+\|x-j\|_{2}),\\
&\psi_j(k;x)=\begin{cases}
\phi_j(k;x) \text{ if $j\notin B_x$},\\
-\frac{k}{\beta+ba}\log(1 +\frac{|k|}{c(\beta+ba)}) \text{ otherwise},
\end{cases}\\
&\rho_j(k)= (\psi_j(k;m)-\psi_0(k;m))-(\psi_j(k;\ell)-\psi_0(k;\ell)).
\end{align}
Then
\[
\|\nabla\rho\|_{\infty}\leq \sqrt{\frac{b-1}{b}}
\]
and, for this choice of $\rho$ we have 
\begin{align}
&k[\rho_m-\rho_\ell]=\frac{2k^2}{\beta+ba}\left[\log(1+\|m-\ell\|_{2})-\log\left(1 +\frac{2 |k|\sqrt{b}}{\sqrt{b-1}(\beta+ba)}\right)\right],\\
&\sum_{(jj')\in \Lambda } [\beta+ba] (\rho_j-\rho_j')^2/2\leq \frac{1}{2} k[\rho_m-\rho_\ell],
\end{align}
so that 
\[
\left|\left \langle e^{ik(t_m-t_\ell)} \right \rangle_{\Lambda, a, J, \varepsilon_0} \right|\leq \exp\left(-\frac{k^2}{\beta+ba}\left[\log(1+\|m-\ell\|_{2})-\log\left(1 +\frac{2|k|\sqrt{b}}{\sqrt{b-1}(\beta+ba)}\right)\right]\right)
\]
for any $b>1$.  We take $b=2$ to obtain one of the bounds claimed in the statement of the theorem.

To obtain the second bound, we simply scale the approximate optimizer,
setting
\[
\rho_j=-\frac{k}{4|k|}\log(1+\|m-j\|_{2})/(1+\|\ell-j\|_{2})+\frac{k}{4|k|}\log(1+\|m\|_{2})/(1+\|\ell\|_{2}).
\]
Then $\|\nabla \rho\|_{\infty}\leq \frac{1}{2}$ and, reasoning as above to deal with the cosine terms, we have
\[
\left|\left \langle e^{ik(t_m-t_\ell)} \right \rangle_{\Lambda, a, J, \varepsilon_0} \right|\leq \exp\left(-\frac{[|k|-\beta-a]}{2}\left[\log(1+\|m-\ell\|_{2})\right].\right)
\]

\section{Bound on Laplace Transforms: Proof of Theorem \ref{T:Main2} }
\label{S:Laplace}
Recall now \Cref{L:Laplace}.  Properly speaking, the proof of this  lemma requires an \textit{a priori}  identity to the effect that the $2a^{th}$ moment of $e^{t_v}$ is $1$.  The latter identity was already proved for the equivalent of one point pinning in \cite{S, H02}.  However, we feel it is worth pointing out that if $a\in \N+1/2$ then the identity is \textit{algebraic} and therefore applies to any choice of pinning:
\begin{lemma}
\label{L:Moment}
For any finite weighted graph $(G, J)$ and any pinning $\eps$, 
\[
\langle e^{2a t_v} \rangle_{G, J, a, \eps}=1.
\]
\end{lemma}
\begin{proof}
Since $z_v+x_v=e^{t_v}$,
\[
\langle e^{2a t_v} \rangle_{J, \eps}=\sum_{j=0}^{2a} {2a \choose j}\langle z_v^j x_v^{2a-j} \rangle_{J, \eps}=\sum_{j=0, j \text{ odd}}^{2a} {2a \choose j} \langle z_v^j x_v^{2a-j} \rangle_{J, \eps}.
\]
The restriction to $j$ odd on the RHS is due to the fact that $\langle z_v^j x_v^{2a-j} \rangle_{J, \eps}=0$ if $2a-j$ is odd.
Using $a-j/2$ distinct supersymmetry transformations,
\begin{multline}
\langle z_v^j x_v^{2a-j} \rangle_{J, \eps}=\frac{\Gamma(2a-j+1)}{2^{a-j/2}\Gamma(a-j/2+1)} \left \langle z_v^j \prod_{\ell=1}^{a-j/2}[- \xi_v^{\ell}\eta_v^{\ell}]\right \rangle_{J, \eps}\\
=\frac{\Gamma(2a-j+1)}{2^{a-j/2}\Gamma(a-j/2+1)} \left \langle \sigma_v^j \prod_{\ell=1}^{a-j/2}[- \psibar_v^{\ell+(j-1)/2}\psi_v^{\ell} ]\right \rangle_{J, \eps}^f.
\end{multline}

Now we note that if we define
\[
\tilde{\sigma}_v(j)^2=1+2\sum_{\ell=a-j/2}^{a-1/2}\psibar_v^{\ell}\psi_v^{\ell}=:1-2\alpha_j
\]
the nilpotency of the Grassmann variables implies
\[
\sigma_v^j \prod_{\ell=1}^{a-j/2}[- \psibar_v^{\ell}\psi_v^{\ell} ]= \tilde{\sigma}_v(j)^j \prod_{\ell=1}^{a-j/2}[- \psibar_v^{\ell+(j-1)/2}\psi_v^{\ell+(j-1)/2} ].
\]
Expanding $ \sigma_v(j)^j$, we have
\[
 \sigma_v(j)^j =\sum_{k=0}^{(j-1)/2}\frac{(j/2)\cdots (j/2-k)(-2)^k\alpha_j^k}{k!} 
\] 
By exchangeability, after we insert this expression,
\[
\langle e^{2a t_v} \rangle_{J, \eps}=\sum_{L} d_L \left \langle\prod_{\ell=1}^{L}[-\psibar_v^{\ell}\psi_v^{\ell}] \right \rangle_{J, \eps}^f,
\]
where
\begin{align*}
d_{L}=&2^{-L}\sum_{\substack{k, j: \\ j\text{ odd}, j\leq 2a ,\\ k \leq (j-1)/2,\\ a-j/2+k=L}} \frac{(2a)!}{j! (a-j/2)!}\frac{(j/2)\cdots (j/2-k)(-1)^k2^{2k}}{k!}{(j-1)/2 \cdots [(j-1)/2-k]}\\
=&2^{-L}\sum_{k=0}^L \frac{(2a)!(-1)^k}{4(L-k)!k!(2a_0-L)!}\\
=&0.
\end{align*}
The claim now follows.
\end{proof}

\begin{proof}[Proof of Theorem \ref{T:Laplace}]
Given $\rho$ as in Lemma \ref{L:Laplace}, as long as $\tilde{\beta}_{jj'} \geq 1/2$,  Theorem \ref{T:Main2} implies
\[
\langle e^{s 2a [t_v-t_0]} \rangle_{\Lambda, \beta, a, \eps_0}   \leq  e^{-2as\gamma +(\beta+1)  q^2\gamma^2 \CE(\rho)}
\]
where 
\[
 \CE(\rho)=\sum_{(jj')}  (\rho_j-\rho_{j'})^2
\]
is the Dirichlet energy of $\rho$.  

Now choose $\rho_j$ to be the function which is $0$ at $0$, $1$ at $v$ and discrete harmonic on $\Lambda \backslash \{0,v\}$.  As is well know, if $\Lambda$ is large enough (for $y$ fixed but large), we can find $c_0$ so that
\[
\|\nabla \rho\|_{\infty}\leq  \frac{c_0}{\log |v|}
\]
\[
\CE(\rho)\leq \frac{c_0}{\log |v|}
\]
and
The constants $c_0$ and $c_1$ are independent of $\Lambda$ as $\Lambda \uparrow \Z^2$, provided the surface-to-volume ratio  $\frac{|\partial \Lambda|}{|\Lambda|}$ tends to $0$.
To finish, given $s$, we optimize over $\gamma$ subject to  the hypothesis of \Cref{L:Laplace} is satisfied.  Thus we set
\[
\gamma:=\frac{as \log |v|}{c_1(\beta+1)q^2}
\]
where $c_1=\max( \frac{2as c_0}{(\beta+1)}, 1)$.  With this choice $q^2 \gamma\|\nabla \rho\|_{\infty}\leq 1/2$ and
\[
-2as\gamma +(\beta+1)  q^2\gamma^2 \CE(\rho)\leq -as\gamma=-c_2 \log|v|. 
\]
where $c_2 \frac{(as)^2}{c_1( \beta+1)q^2}$.
The theorem thus holds.
\end{proof}

\appendix

\section{SUSY for $a\in \BN+1/2$}
\label{A:SUSY}
In this appendix we sketch the development of the $\H^{2|2N}$ nonlinear $\sigma$-model from which the measure $d\mu_{G, a, J, \varepsilon} (t)$, defined at \eqref{E:t-meas} is derived.  Let us first introduce the $\H^{2|2N}$ $\sigma$-models.  We consider the general case where the model is defined over a finite graph $G=(V, E)$.

For each $j \in V$ we
introduce a supervector $u_j \in \mathbb{R}^{3|2N}$,
\begin{equation}
    u_j = (z_j, {x}_j, y_j,  {\xi}_j, {\eta}_j).
\end{equation}
The variables ${\xi}_j, \eta_j\,$ are $N$-tuples of odd generators for a Grassmann algebra over $j$. Subscript indices denote locations in $G$ whereas  supercript indices will denote internal components of the Grassmann vectors:  $\xi^{(\ell)}_i$ is the $\ell$'th component of $\xi_i$.  We then define a Minkowski signature bilinear form on
$\mathbb{R}^{3|2N}$ by
\begin{equation}\label{eq:inprod}
    (u,u') = - z z' +x x'+y y' + \xi\cdot \eta' - \eta \cdot \xi'
    \end{equation}
where  
\begin{align}
\xi\cdot \eta'= \sum_{\ell=1}^N \xi^{\ell} {\eta'}^{\ell}.
\end{align}
To define the $\H^{2|2N}$ $\sigma$-model, the $u_j$'s are constrained  to satisfy the quadratic equation
\begin{equation}
  (u_j\, , u_j) = - 1.
\end{equation}
Note that the solutions to this equation form a two-sheeted hyperboloid in $\mathbb{R}^{3|2N}$. 
If, for each $j$, we restrict attention to spins lying in the sheet with positive square root, we arrive at
$\mathrm{H}^{2|2N}$.  This is a super manifold.  Geometrically, it is an infinitesimal extension by Grassmann variables of the $d=2$ hyperbolic plane, parametrized by $2$ real variables $x_j, y_j$ 
and $2N$ Grassmman variables $\xi_j, \eta_j$.

On the product space $(\mathrm{H}^{2|2N})^{|\Lambda|}$ we introduce a
`measure' (more accurately, a Berezin superintegration form)
\begin{equation}\label{eq:dmu}
    D\mu_\Lambda = \prod_{j\in V}  \textrm{d}x \textrm{d}y
   \prod_{j\in V, \ell} \partial_{\xi^{\ell}_j} \partial_{\eta^{\ell}_j} \circ (z_j)^{-1/2} \;.
\end{equation}
Note that $z_j^2=1+x_j^2+y_j^2 + 2\xi_j\cdot \eta_j$ - it is an element of the Grassmann algebra, not a real variable.
 The Gibbs state is then proportional to the Grassmann integration form $D\mu_\Lambda \,\mathrm{e}^{-A_{J, \varepsilon}}$ with
\beq
    A_{J,\varepsilon} = \frac{1}{2}\mathop\sum\limits_{(jj')\in E}
    J_{jj'}\,(u_j-u_j'\,,u_j-u_j') +  \mathop\sum\limits_{i
    \in V} \varepsilon_i(z_i-1) \label{eq:fullaction} \\
  \eeq
The coupling constants $J_{jj'} >0$ if $jj'$ are nearest neighbors in G and $J_{jj'} =
0$ otherwise. 

\subsection{Horospherical Coordinates and $d\mu_{G, a, J, \varepsilon} (t)$}
\label{S:Horo}
To connect the $\H^{2|2N}$ model to  $d\mu_{G, a, J, \varepsilon} (t)$ with $a=N-1/2$ we need to introduce a change of variables called horospherical (or Iwasawa) coordinates.  The next coordinates are $(t, s, \bar{\theta}, \theta)$ defined by
\begin{equation}\label{eq:horoc}
    x = \sinh t - \mathrm{e}^t \left(  {\textstyle{\frac{1}{2}}}
    {\bf{s}}\cdot {\bf{s}} + \sum_{i=1}^N\bar\theta^{\ell}\theta^{\ell} \right)\ , \quad y = \mathrm{e}^t s \text{ for $i\geq 2$}\ ,\quad
    \xi^{\ell} = \mathrm{e}^t \bar\theta^{\ell}\;, \quad \eta^{\ell} = \mathrm{e} ^t \theta^{\ell}\;,
\end{equation}
where $t\in \BR$ and $s\in \BR$ range over the real numbers. In the Poincar\'{e} disc and given a point $p$, $t$ represents the signed distance of the horocycle tangent to $1$ from  $0$ and containing $p$ whereas $s$ is represents the (normalized) location of $p$ on the horocycle). In these coordinates, the expression for the action becomes
\begin{equation}\label{eq:action}
    A_{J,\varepsilon} =  \mathop\sum\limits_{(ij)}
   J_{ij} (S_{ij}-1) +  \mathop\sum\nolimits_{k\in\Lambda}
    \varepsilon_k (z_k - 1) \;,
\end{equation}
where $(jj')$ are NN pairs and
\begin{align}
    S_{jj'} &=B_{jj'} + (\bar\theta_j - \bar\theta_{j'})\cdot (\theta_j -
    \theta_{j'})\, \mathrm{e}^{t_j + t_{j'}}\;, \label{eq:def-S}\\
    B_{jj'} &= \cosh (t_j-t_{j'}) + {\textstyle{\frac{1}{2}}}
    (s_j - s_j')\cdot (s_j-s_{j'})\, \mathrm{e}^{t_j + t_{j'}}\;, \label{eq:Bdef}\\
    z_{j} &= \cosh t_j +\left({\textstyle{\frac{1}{2}}}s_j\cdot s_j+
    \bar\theta_j \cdot  \theta_j\right)\mathrm{e}^{t_j} \label{eq:def-z}\;.
\end{align}
By applying Berezin's transformation
formula \cite{berezin} for changing variables in a (super-)integral,
one finds that
\begin{equation}
    D\mu_{V} = \prod\nolimits_{j\in\Lambda} 
    \mathrm{e}^{-t_j[2N-1]} \textrm{d}t_j \textrm{d} s_j \, \prod_{\ell} \partial_{\bar{\theta}^{\ell}_j}
    \partial_{\theta^{\ell}_j} \;
\end{equation}

For any function $f$ of the  field variables $\{ t_j\,, s_j\,,
\bar\theta_j\,,\theta_j\}_{j\in\Lambda}$ we now define its expectation as
\begin{equation}\label{expectation}
     \left\langle f  \right \rangle_{G, J,\varepsilon}  =
   \frac{ \int D\mu_\Lambda \, \mathrm{e}^{-A_{J,\varepsilon}} f}{\int D\mu_V \, \mathrm{e}^{-A_{J,\varepsilon}} } \;,
\end{equation}
whenever the numerator integral exists.  

There is no notational conflict with earlier definitions due to the following.  Observe the  beautiful feature that all coordinates, besides the $t$ coordinate, are Gaussian in this representation.
Thus if $f$ depends on the $t_j$'s alone, we  may easily integrate the other variables out.  We find ($a=N-1/2$)
\begin{align*}
&\int D\mu_\Lambda \, \mathrm{e}^{-A_{J,\varepsilon}}=Z_{G, a, J, \eps}
&\left\langle f \right \rangle_{G, J,\varepsilon}= \frac{\int f  d\mu_{\Lambda }^{J, \varepsilon} (t) }{Z_{G, a, J, \eps}}=\left\langle f \right \rangle_{G, a,  J,\varepsilon}.
\end{align*}

\subsection{SUSY Localization}
\Cref{S:MS} highlighted the importance of controlling partition function ratios
\[
R_{J', J}:=\frac{Z_{G, a, J', \eps}}{Z_{G, a, J, \eps}}
\]
for two choices of coupling constants $J, J'$.  For the $\sigma$-models taking values in $\H^{2|2}$, this ratio is $1$ by SUSY localization \cite{DSZ}.  On $\H^{2|2N}$ this is not true, but the localization argument is still extremely useful.

We now sketch the localization computation from \cite{DSZ} as it applies to the $\H^{2|2N}$ models (the reader may consult that paper for the missing details).   The computation is most easily explained by first considering the special case that $G$ is just a single vertex.  Let $H$ be the quadratic polynomial
\begin{displaymath}
    H = x^2 +y^2 +  2 \xi \cdot \eta.
\end{displaymath}
Let us isolate one pair of Grassmann components $(\xi^1, \eta^1)$.  With respect to this pair, let  $q$ be the distinguished first-order differential operator
defined by
\begin{equation}\label{eq:def-Q}
    q = x \partial_{\eta^{(1)} } - y \partial_{\xi^{(1)} }
    + \xi^{1}  \partial_{x}  + \eta^{1}  \partial_{y}  \;.
\end{equation}
Note that $q$ annihilates $H$,

In this notation, the \textit{a priori}
superintegration form is
\begin{displaymath}
    D\mu = \textrm{d}{ x} \textrm{d} y \, \prod_{\ell} \partial_{\xi^\ell} \partial_{\eta^\ell}
    \circ (1 + H)^{-1/2} \;.
\end{displaymath}
\begin{lemma}\label{lem:inv-int}
The Berezin superintegration form $D\mu$ is $q$-invariant, i.e.,
\begin{displaymath}
    \int D\mu\; q \cdot f = 0
\end{displaymath}
for any compactly supported smooth superfunction $f$.
\end{lemma}

\begin{corollary}
Suppose $q\cdot f=0$.
Then for any $\tau>0$,
\[
 \int D\mu\;  f=  \int D\mu\; e^{-\tau H} f.
\]
\end{corollary}

Every superfunction $f$ can be expanded over the Grassmann variables $(\xi^{\ell}, \eta^{\ell})_{\ell \geq 2}$ as
\[
f=\sum_{I,J\subseteq \{2, \dotsc, N\}} f_{I,J}(x,y, \xi^{1}, \eta^{1}) \xi^{I}\eta^{J}
\]
where $\xi^I=\prod_{\ell\in I} x^\ell$ and simialrly for $\eta$.
Let $f_0$ be the superfunction obtained by setting $x = y=\xi^{(1)} = \eta^{(1)}=0$,
\[
f_0=\sum_{I,J\subseteq \{2, \dotsc, N\}} f_{I,J}(o) \xi^{I}\eta^{J},
\]
so that the coefficient functions are evaluated as $ f_{I,J}(o)$.  Thus superfunction is a constant coefficient polynomial in $(\xi^{\ell}, \eta^{\ell})_{\ell \geq 2}$.
Let $n=N-1$ and let $\bar{\psi}, {\psi}$ denote the $n$-tuples defined by the last $n$ components of $\xi, \eta$, so $\bar{\psi}^{\ell} =\xi^{\ell+1}, \psi^{\ell}=\eta^{\ell+1}$ for $\ell\in \{1, \dotsc, n\}$ and let $\sigma=(1 + 2\bar{\psi}\cdot {\psi})^{1/2}$.  The variables $\sigma, \psibar, \psi$ live in a degenerate hyperboloid $\H^{0|2n}\subset \BR^{1|2n}$.

\begin{displaymath}
    D\mu_0 =  \prod_{\ell} \partial_{\bar{\psi}^{\ell}} \partial_{{\psi}^{\ell}}
    \circ \sigma^{-1} \;.
\end{displaymath}

\begin{lemma}[Localization from $\H^{2|2N}$ to $\H^{0|2(N-1)}$]
\label{lem:SUSYloc}
Let $f$ be a smooth superfunction which satisfies
the invariance condition $Q f = 0$ and decreases sufficiently fast at
infinity in order for the integral $\int D\mu\, f$ to
exist. Then
\begin{displaymath}
    \int D\mu\, f = \int D\mu_0 f_0.
\end{displaymath}
\end{lemma}

We now generalize this last lemma to ${\H^{2|2N}}^V$ for general finite graphs.
Let
$\n_j=(\sigma_j, \bar{\psi}_{j}, \psi_{ j'})$
and set
\beq
   S_{J,\varepsilon} = \frac{1}{2}\mathop\sum\limits_{(jj')\in E}
    J_{jj'}\,(\n_j-\n_{j'}\,,\n_j-\n_{j'}) +  \mathop\sum\limits_{i
    \in \Lambda} \varepsilon_i(\sigma_i-1) \label{eq:fullaction1} \\
  \eeq
where the bilinear form is
\[
(\n,  \n')=-\sigma\sigma' + \psibar \cdot \psi'+ \psibar'\cdot \psi.
\]
Let $\left \langle \cdot \right \rangle^f_{G, n, J, \epsilon}$ denote the corresponding Gibbs state, in particular
\[
Z^f_{G, n, J, \eps}:=\int \prod_{j\in V} D\mu_{0}(\psibar_j, \psi_j) e^{-S_{G, n, J, \epsilon}}
\]
Let $q_{V}=\sum_{j \in V} q_j$
  
\begin{lemma}\label{lem:SUSYloc1}
Let $N\in \BN$, $a=N+1/2$, $n=N-1$.  Then the $\H^{2|2N}$ $\sigma$-model is equivalent to a purely fermionic $\H^{0|2n}$ $\sigma$-model in the sense that for any function $F(\underline{z}, (\underline{\xi}^\ell)_{\ell=2}^N,(\underline{\eta}^\ell)_{\ell=2}^N)$ which decays sufficiently fast and such that $q_{V} \cdot F=0$
\beq
\label{E:Part}
\left \langle F(\underline{z}, (\underline{\xi}^\ell)_{\ell=2}^N,(\underline{\eta}^\ell)_{\ell=2}^N) \right \rangle_{G, N, J, \epsilon}= \left \langle F (\sigma, 0,\bar{\psi}, \psi) \right \rangle^f_{G, n, J, \epsilon}.
\eeq
In particular
\[
Z_{G, a, J, \epsilon}=Z^f_{G, n, J, \eps}\]
\end{lemma}

\subsection{Residual SUSY After Localization; Connection with \cite{Sokal-etal}}
\label{S:R12n}
The action $S_{J, \eps}$ can be interpreted as a nonlinear $\sigma$-model with respect to the target space $\BR^{1|2n}$, with even coordinate $\sigma$ and odd generators $(\psibar^{\ell}, \psi^{\ell})_{\ell=1}^n$.  There are two natural quadratic forms we can put on this space:
\begin{eqnarray}
\label{E:hyp}-\sigma^2 +2\psibar\cdot \psi & \text{ Lorentizian},\\
\label{E:sph} \sigma^2+ 2\psibar\cdot \psi & \text{ Euclidean}.
 \end{eqnarray}
Constraining spins to be one when evaluated by the quadratic form \Cref{E:hyp} gives a nonlinear $\sigma$-model with target space one sheet of the degenerate hyperboloid $\sigma^2- 2\psibar\cdot \psi=1$ whereas if we use the quadratic form \Cref{E:sph},  the spins are interpreted as taking values in the 'upper hemisphere' of the degenerate sphere $\sigma^2+2\psibar\cdot \psi=1$. The latter situation was discussed in \cite{Sokal-etal}.  However, these two superspaces are the same via a change of fermionic coordinates.  As such, the discussion of Section $7$ in \cite{Sokal-etal} provides useful insight here.  
Let us adapt and generalize that discussion for the sake of completeness.  

We begin by introducing, at each vertex $i \in V$,
a superfield $v_i := (\sigma_i,\psi_i,\psibar_i)$
consisting of a single bosonic variable $\sigma_i$ 
and 2n Grassmann variables $(\psi^{\ell}_i, \psibar^{\ell}_i)_{\ell=1}^{n}$.
We equip $\BR^{1|2n}$ with the Lorentzian scalar product
\begin{equation}
 (  v_i, v_j)
   \; := \;
- \sigma_i \sigma_j +(\psibar_i \cdot \psi_j - \psi_i \cdot \psibar_j)
   \;,
 \label{eq.scalarprod}
\end{equation}

There are two types of symmetries which preserve this bilinear form.
The first type is a symplectic linear transformation mapping the Grassmann variables into themselves and fixing the bosonic component.  That is, if
$M$ 
is an invertible $2n$-by-$2n$ matrix preserving the blilinear form induced by 
\[
J:=\left(\begin{array}{cc}0 & I_n \\-I_n & 0\end{array}\right) 
\]  
and if $u_i =(\sigma_i, M\cdot [\psi_i, \bar{\psi}_i])$, then $(u_i, u_j)=(v_i, v_j)$.

The second type of transformation is supersymmetric, mixing $\sigma$ with the $\psi, \bar\psi$'s.  There are, in the general case, $n$ \textit{noncommuting} SUSY transformations transformations, parametrized by
fermionic (Grassmann-odd) global parameters $(\epsilon^{\ell},\bar{\epsilon}^{\ell})_{\ell=1}^n$:
\begin{eqnarray}
 \delta \sigma_i & = &
    (\bar{\epsilon}^{\ell} \psi_i^{\ell}  +  \psibar_i^{\ell} \epsilon^{\ell})  \\
 \delta \psi_i^{\ell} & = &
    \epsilon^{\ell} \,\sigma_i     \\
 \delta \psibar_i^{\ell} & = &
    \bar{\epsilon}^{\ell} \,\sigma_i
 \label{def.supersym}
\end{eqnarray}
To check that these transformations leave \cref{eq.scalarprod} invariant,
we compute
\begin{eqnarray}
   \delta( v_i\cdot v_j )
   & = &
   (\delta \sigma_i) \sigma_j + \sigma_i (\delta \sigma_j)
   +  \big[ (\delta \psibar_i)\cdot \psi_j + \psibar_i\cdot (\delta \psi_j)
                  -(\delta \psi_i) \cdot \psibar_j - \psi_i\cdot (\delta \psibar_j) \big]
      \nonumber \\ \\[-1mm]
   & = &
   -(\bar{\epsilon}^{\ell} \psi_i +  \psibar_i \epsilon^{\ell}) \sigma_j
   - (\bar{\epsilon}^{\ell} \psi_j +  \psibar_j \epsilon^{\ell}) \sigma_i
           \nonumber \\
   &  & \qquad
   +\, \big[ \bar{\epsilon}^{\ell} \psi^{\ell}_j \sigma_i
                          +  \psibar^{\ell}_i \epsilon^{\ell} \sigma_j
              -\epsilon^{\ell} \psibar^{\ell}_j \sigma_i - \psi^{\ell}_i \bar{\epsilon}^{\ell} \sigma_j \big]
      \\[2mm]
   & = &
   0  \;.
\end{eqnarray}

Now let us consider a $\sigma$-model in which the superfields $v_i$
are constrained to lie on the upper sheet of the hyperboloid $\BR^{1|2n}$,
$
\sigma_i^2  -2 \psibar_i \cdot \psi_i  \;=\; 1
$
This constraint is solved by writing
\begin{equation*}
\sigma_i  \;=\;  \pm (1 + 2 \psibar_i\cdot  \psi_i)^{1/2},
 \label{def.sigmai}
\end{equation*}
so that $\sigma_i$ is an even invertible element of the Grassmann algebra.
We take only the $+$ sign in \eqref{def.sigmai} and denote the corresponding unit vector by $\n_i$.

The $\ssp(2n)$ transformations continue to act as above
while the SUSY transformations act via 
\begin{eqnarray*}
 \delta \psi^{\ell}_i & = &
   \epsilon^{\ell} \sigma_i   \\
 \delta \psibar^{\ell}_i & = &
  {\bar{\epsilon}}^{\ell}  \sigma_i  
   \label{eq.supersym.psionly}
\end{eqnarray*}
These transformations leave invariant the scalar product $  \n_i\cdot \n_j$.
and the corresponding generators $Q^{\ell}_\pm$ are defined as
\begin{eqnarray*}
  Q^{\ell}_+  & = &
    \sum_{i\in V}  \sigma_i \partial^{\ell}_i 
        \\[1mm]
  Q^{\ell}_-  & = &  \sum_{i\in V}  \sigma_i       \bar{\partial}_i 
     \label{eq:def_Q_bis1}
\end{eqnarray*}
where 
$\partial^{\ell}_i = \partial_{\psi_i^{\ell}}$
and $\bar{\partial}^{\ell}_i = \partial_{\psibar_i^{\ell}}$
\begin{eqnarray}
\label{E:SUSY4}
&Q^{\ell}_{+}{ {\psi}}_i^\ell=Q^{\ell}_{-}{ {\bar{\psi}}}_i^\ell= \sigma_i, &Q^{\ell}_{\pm}[ \n_i\cdot \n_j]=0.
\end{eqnarray}
From this identity,  $Q^{\ell}_{\pm} S_{J, \eps}=Q^{\ell}_{\pm} e^{-S_{J, \eps}}=0$.  Note also  that \textit{a priori} integration form $\CD_0(\psibar, \psi)$ is invariant with respect to the $Q^{\ell}_{\pm}$'s:
\[
\int \CD_0(\psibar, \psi) Q^{\ell}_{\pm} F(\psi, \bar{\psi})=0.
\]
From these facts, \eqref{E:IBP} follows immediately.
\bibliography{Mc-B}
\bibliographystyle{alpha}

\medskip{}

$~$\\
Department of Mathematics, \\
The Technion; Haifa, Israel.\\
E-mail: nickc@technion.ac.il\\

\end{document}